%% file: main.tex
\documentclass[runningheads]{llncs}
\usepackage{csquotes}
\usepackage{graphicx} 
\usepackage[utf8]{inputenc}
\usepackage{todonotes}
\usepackage{nimbusmononarrow} 
\usepackage[scr=boondoxo]{mathalpha} 
\usepackage{subcaption}
\usepackage{xspace}
\usepackage{hyperref}
\usetikzlibrary{decorations.pathmorphing}
\usetikzlibrary{shapes.geometric}
\usetikzlibrary{calc}

\usepackage[style=alphabetic]{biblatex}
\usepackage{mathtools}

\usepackage{amssymb}
\newtheorem{conj}{Conjecture}
\newtheorem{observation}{Observation}

\begin{document}
\include{notation}
\title{Global One-Counter Tree Automata}
%
%
\author{Luisa Herrmann\inst{1,2} \and
Richard Mörbitz\inst{1}}
\authorrunning{L. Herrmann, R. Mörbitz}
%
\institute{Faculty of Computer Science, TU Dresden, Germany \and
Center for Scalable Data Analytics and Artificial Intelligence (ScaDS.AI) Dresden/Leipzig, TU Dresden, Germany\\
\email{\{luisa.herrmann,richard.moerbitz\}@tu-dresden.de}}
\maketitle              
\begin{abstract}
    We introduce global one-counter tree automata (GOCTA)
    which deviate from usual counter tree automata
    by working on only one counter which is passed through the tree
    in lexicographical order,
    rather than duplicating the counter at every branching position.
    We compare the capabilities of GOCTA to those of counter tree automata
    and obtain that their classes of recognizable tree languages are incomparable.
    Moreover, we show that the emptiness problem of GOCTA is undecidable while,
    in stark contrast, their membership problem is in P.

\end{abstract}

\section{Introduction}

Similar to the case of finite string automata, there is a long tradition of adding counting mechanisms to finite tree automata in order to increase their expressiveness, resulting in models like \emph{pushdown tree automata} \cite{Gue81}, \emph{one-counter tree automata} (OCTA), or \emph{Parikh tree automata} \cite{KlaRue03,Kla04}.
OCTA are a special case of
pushdown tree automata
where the pushdown alphabet consists of only two symbols,
one of which being the bottom marker enabling zero-tests during a computation. Thus, OCTA generalize one-counter string automata (OCA)
by allowing branching in the input. Although not being explicitly described in the literature, they occur as 
special case of \emph{regular tree grammars with storage} \cite{Eng86} and are rather folklore.
With their classical definition, OCTA read the input top-down
where independent subtrees are processed in parallel.
At each instance of branching, the counter is duplicated
and copies are passed to the successors (cf.\@ the left of Fig.\;\ref{fig:storage-flow}).
Thus, in contrast to OCA, there is no path among which one instance of the counter
traverses the entire input.

On the other hand, Parikh tree automata (PTA) count on their input in a global manner by adding up integer vectors (assigned to symbols) over the whole tree and testing once in the end of the computation if the sum is contained in a semilinear set. While this view of the entire input allows to recognize the language $T_{a=b}$ of all trees where two fixed symbols $a$ and $b$
must have the same number of occurrences, OCTA are not able to do the same.
In addition to this global view of the storage,
PTA also differ from OCTA by allowing multiple counters, but they lack the capability to test their counters for $0$ during a run.

\cite{Castano04} introduced a global storage flow for indexed grammars,
which are essentially context-free grammars enhanced by a pushdown,
introducing global index grammars (GIG).
He showed that GIG are suitable for natural language processing (NLP)
since they have several properties which are desirable in this area:
their membership problem (called parsing in NLP) is in P
and their emptiness problem is decidable.
To prove the latter statement, he showed that their generated languages are semilinear. 

Inspired by GIG and the global counting mode of Parikh tree automata, we introduce global one-counter tree automata (GOCTA).
They are essentially OCTA where the counter is passed through the input tree
in the lexicographical order of its positions (cf.\@ the right of Fig.\;\ref{fig:storage-flow}).
Thus, a single instance of the counter reaches the entire input.
Indeed, we find that GOCTA can recognize  $T_{a=b}$.
As the first main contribution of this paper,
we explore the expressiveness of GOCTA and
we compare the classes of tree languages recognizable by OCTA and by GOCTA,
respectively;  we find that they are incomparable (Corollary~\ref{cor:incomparable}).
Second, we obtain that the emptiness problem of GOCTA is undecidable
(Theorem~\ref{thm:emptiness}) which casts doubt on the corresponding result for GIG
by~\cite{Castano04}.
Finally, we show that the membership problem of GOCTA
is decidable in polynomial time with respect to the size of the input tree (Theorem~\ref{thm:membership}).

\begin{figure}[t]
    \centering
    \includegraphics[width=0.9\textwidth]{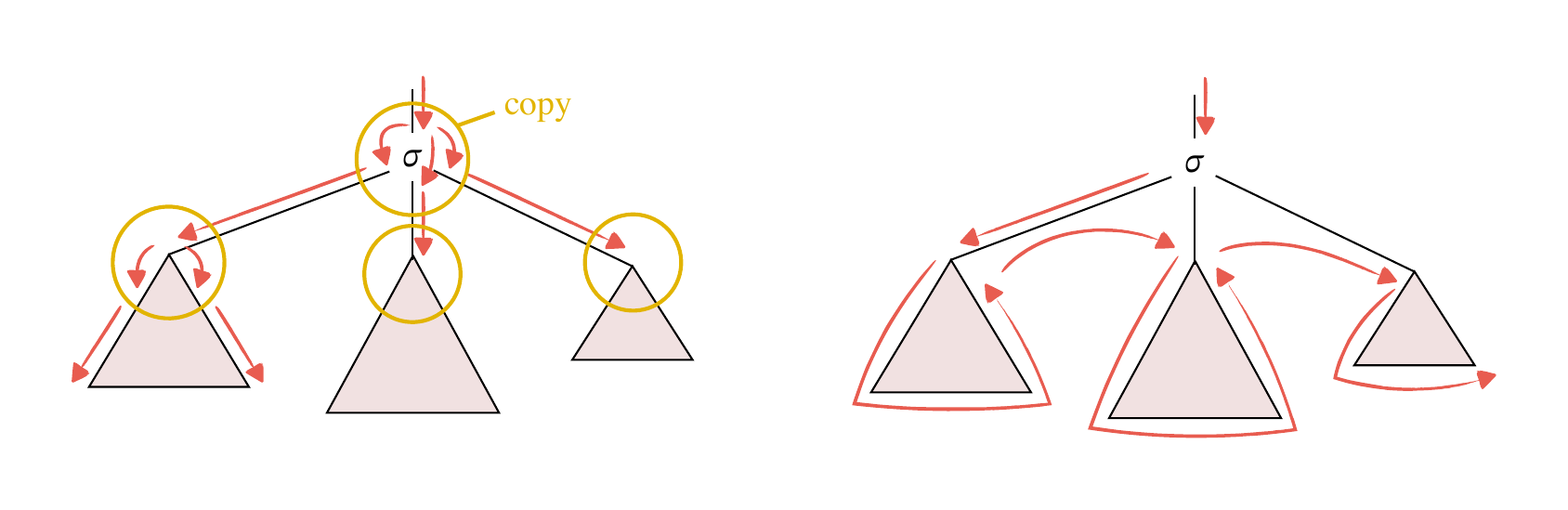}
    \caption{Storage flow of OCTA (left) vs.\@ GOCTA (right).}
    \label{fig:storage-flow}
\end{figure}

\section{Preliminaries}

In the usual way, we denote by $\Int$ and $\Nat$ the \emph{set of integers} and \emph{natural numbers}, respectively, and use $[i,j]$ to denote the interval $\{\ell\in\Int\mid i\leq \ell\leq j\}$ for $i,j\in\Int$. For each $n\in\Nat$ we set $[n]=[1,n]$.

Given a finite set $A$, we mean by $|A|$ the number of its elements.
For each $k \in \Nat$ and $a_1, \dots, a_k \in A$,
we call $w = a_1 \cdots a_k$ a \emph{a string (of length $k$) over $A$}
and, for each $i \in [k]$, we let $w[i]$ denote $a_i$.
The set of all strings of length $k$ over $A$ is denoted by $A^k$
and we let $A^* = \bigcup_{k \in \Nat} A^k$.
We note that $A^0 = \{ \varepsilon \}$ where $\varepsilon$ is the empty string.
We let $\lex$ denote the usual lexicographical order on $\Nat^*$.

Given a finite set $\Sigma$ and a mapping $\rk\colon\Sigma\to\Nat$, we call the tuple $(\Sigma,\rk)$ a \emph{ranked alphabet}. Usually we only write $\Sigma$ and assume $\rk$ implicitly. We denote, for each $n\in\Nat$, $\rk^{-1}(n)$ by $\Sigma^{(n)}$ and by writing $\sigma^{(n)}$ we mean that $\rk(\sigma)=n$.

Let $\Sigma$ be a ranked alphabet and let $H$ be some set. The set $T_\Sigma(H)$ of \emph{trees (over $\Sigma$ and indexed by $H$)} is the smallest set $T$ such that $H\subseteq T$ and $\sigma(\xi_1,\ldots,\xi_n)\in T$ for each $n\in\Nat$, $\sigma\in\Sigma^{(n)}$, and $\xi_1,\ldots,\xi_n\in T$. We write $T_\Sigma$ if $H=\emptyset$. Given $\sigma\in\Sigma^{(n)}$ and $L_1,\ldots,L_n\subseteq T_\Sigma$, we let $\sigma(L_1,\ldots,L_n)=\{\sigma(\xi_1,\ldots,\xi_n)\mid \xi_i\in L_i\}$. Each subset $L\subseteq T_\Sigma$ is called a \emph{tree language}. 

Now let $\xi\in T_\Sigma(H)$.
As usual, $\pos(\xi)\subseteq\Nat^*$ is the set of \emph{positions} of $\xi$, $\xi(\rho)\in\Sigma$ refers to the \emph{label of $\xi$ at position $\rho$}, and we set $|\xi|=|\pos(\xi)|$ standing for the \emph{size of $\xi$}.
For each $\sigma\in\Sigma$ we let $|\xi|_\sigma=|\{\rho\in\pos(\xi)\mid\xi(\rho)=\sigma\}|$.
If $\xi$ has the form $\gamma_1(\dots\gamma_k(\#))$ for $\gamma_i\in\Sigma^{(1)}$, $i\in[k]$, we denote it as the string $\gamma_1 \cdots \gamma_k \#$ of length $k+1$.
Since each string $w$ can be uniquely considered as such a \emph{monadic} tree,
we use notations defined for trees, e.g., $|w|_\sigma$, also for strings.
For each $\xi \in T_\Sigma$, we let $\height(\xi) = \max \{ |\rho| \mid \rho \in \pos(\xi) \}$.

We fix a countable set $X = \{x_1,x_2,\ldots\}$ of \emph{variables} that is disjoint from each ranked alphabet in this work and let $X_n = \{x_1,..., x_n\}$ for each $n \in\Nat$. Now let $k\geq 1$ and $\xi\in T_\Sigma(H\cup X_k)$. We say that $\xi$ is a \emph{context} if (1) for each $i\in[k]$ there is exactly one position $\rho_i\in\pos(\xi)$ with $\xi(\rho_i)=x_i$ and (2) for each $i_1,i_2\in[k]$, if $i_1<i_2$, then $\rho_i\leq_{\mathrm{lex}}\rho_j$. We denote the set of all such contexts as $C_\Sigma(H,X_k)$. The \emph{composition} of a tree $\xi\in T_\Sigma(H\cup X_k)$ with trees $\xi_1,\ldots,\xi_k\in T_\Sigma(H)$, denoted by $\xi[\xi_1,\ldots,\xi_k]$, replaces each occurrence of $x_i$ in $\xi$ by $\xi_i$.

Let $\xi \in T_\Sigma$.
We define $\ang{\xi}_0 = x_1$ and, for every $i \in [|\xi|]$,
we let $\ang{\xi}_i$ be the context obtained from $\ang{\xi}_{i-1}$
by replacing $x_1$ (occurring at position $\rho)$ with $\sigma(x_1,\ldots,x_n)$ if $\xi(\rho)=\sigma\in\Sigma^{(n)}$
(and, if necessary renaming the other variables to ensure that $\ang{\xi}_i$ is a context).
Intuitively, $\ang{\xi}_i$ contains exactly the $i$ lexicographically first nodes of $\xi$ and $\ang{\xi}_{|\xi|} = \xi$.
Moreover, we let $\width{i}\in\Nat$ be the smallest number such that $\xi_i\in C_\Sigma(X_{\width{i}})$.

\section{Global One-Counter Tree Automata}

\subsubsection{Automaton model.}
A \emph{global one-counter tree automaton} (GOCTA) is a tuple $\A=(Q,\Sigma,q_0,\Delta)$ where $Q$ is a finite set of states, $\Sigma$ is a ranked alphabet, $q_0\in Q$ is the initial state, and $\Delta$ is a finite set of transitions of the following two forms:
\begin{align*}
q&\xrightarrow{\p\slash \z} q'\tag{$\varepsilon$-transition}\\
q&\xrightarrow{\p\slash \z} \sigma(q_1,\ldots, q_n)\tag{read-transition}
\end{align*}
where $n\in\Nat$, $\sigma\in\Sigma^{(n)}$, $q,q',q_1,\ldots,q_n\in Q$, $\p\in\{0,{>}0,\top\}$, and $\z\in\Int$. We denote by $\Delta_\varepsilon$ and $\Delta_\Sigma$ the sets of all $\varepsilon$-transitions and read-transitions of $\A$, respectively. Given a transition $\tau=q\xrightarrow{\p\slash \z}w\in \Delta$, we refer to $\z$ by $\textsc{Instr}(\tau)$ and set $\textsc{Instr}(\A)=\bigcup_{\tau\in\Delta}\textsc{Instr}(\tau)$. A transition $q\xrightarrow{\top\slash 0}w$ will simply be abbreviated by $q\to w$.

The semantics of a GOCTA $\A=(Q,\Sigma,q_0,\Delta)$ is defined as follows. We denote by $\confA$ the set $T_\Sigma(Q)\times\Nat$. For each transition $\tau\in\Delta$, we let $\Rightarrow^\tau$ be the binary relation on $\confA$ such that for each $\zeta,\zeta'\in T_\Sigma(Q)$ and $m,m'\in\Nat$ we have
\[(\zeta,m)\Rightarrow^\tau(\zeta',m')\]
if there are $k\geq 1$, $\hat{\zeta}\in C_\Sigma(X_k)$ and $q,q_1,\ldots,q_{k-1}\in Q$ such that $\zeta=\hat{\zeta}[q,\bar{q}]$ (with $\bar{q}=q_1,\ldots,q_{k-1}$) and either
\begin{itemize}
    \item[-] $\tau=q\xrightarrow{\p\slash \z} q'$, $\p(m)$, $\zeta'=\hat{\zeta}[q',\bar{q}]$, $m+\z\geq 0$, and $m'=m+\z$, or
    \item[-] $\tau=q\xrightarrow{\p\slash \z} \sigma(p_1,\ldots, p_n)$, $\p(m)$, $\zeta'=\hat{\zeta}[\sigma(p_1,\ldots,p_n),\bar{q}]$, $m+\z\geq 0$, and $m'=m+\z$
\end{itemize}
where $\top(m)$ holds for all $m\in\Nat$, $0(m)$ iff $m=0$, and ${>} 0(m)$ iff $m\geq 1$.
Thus, transitions are always applied at the lexicographically first position carrying a state.
The \emph{computation relation of $\A$} is the binary relation $\Rightarrow_\A=\bigcup_{\tau\in\Delta}\Rightarrow^\tau$. 

A computation is a sequence $t=\zeta_0\Rightarrow^{\tau_1}\zeta_1\ldots\Rightarrow^{\tau_n}\zeta_n$ (sometimes abbreviated as $\zeta_0\Rightarrow^{\tau_1\ldots\tau_n}\zeta_n$) such that $n\in\Nat$, $\zeta_0,\ldots,\zeta_n\in\confA$, $\tau_1,\ldots,\tau_n\in \Delta$, and $\zeta_{i-1}\Rightarrow^{\tau_i}\zeta_i$ for each $i\in[n]$. We call $t$ \emph{successful on $\xi\in T_\Sigma$} if $\zeta_0=(q_0,0)$ and $\zeta_n=(\xi,m)$ for some $m\in\Nat$; the set of all successful computations of $\A$ on $\xi$ is denoted by $\compA{\xi}$. We set $\maxcnt(t)=\max\{m\in \Nat\mid \zeta_i=(\hat\zeta,m), i\in[n]\}$.

\begin{remark}
Given a tree $\xi\in T_\Sigma$, each computation $t\in\compA{\xi}$ is of the form
\begin{align*}
(q_0,0)\Rightarrow^{\theta_1\tau_1}(\ang{\xi}_1[q_1,\bar{q_1}],m_1)\ldots&\Rightarrow^{\theta_{n-1}\tau_{n-1}}(\ang{\xi}_{n-1}[q_{n-1},\bar{q}_{n-1}],m_{n-1})\\
&\Rightarrow^{\theta_n\tau_n}(\ang{\xi}_n,m_n)
\end{align*}
where $n=|\xi|$, $\theta_1,\ldots,\theta_n\in\Delta_\varepsilon^*$, $\tau_1,\ldots,\tau_n\in\Delta_\Sigma$, $m_1,\ldots,m_n\in\Nat$, and for each $i\in[n-1]$ we have $q_i\in Q$ and $\bar{q}_i=p_1,\ldots,p_{\width{i}-1}$ with $p_1,\ldots,p_{\width{i}-1}\in Q$. \exqed
\end{remark}

\begin{remark}
    We note that computations of GOCTA may at first glance look more like derivations of grammars than like runs of automata, since we do not carry the complete input in each computation step for the sake of readability. However, this notational peculiarity can be easily remedied. Furthermore, a slight difference between our model and grammars is that, in our case, at most one terminal is produced per computation step. This is why we speak of \enquote{automata} instead of \enquote{grammars}.
    \exqed
\end{remark}

The \emph{language recognized by $\A$} is the set $\L(\A)=\{\xi\in T_\Sigma\mid \compA{\xi}\neq\emptyset\}$. We say that a tree language $L\subseteq T_\Sigma$ is recognizable by a GOCTA, if there exists a GOCTA $\A$ with $\L(\A)=L$.

We obtain the well-known \emph{finite-state state tree automata} (FTA)
as a special case of GOCTA:
a GOCTA is an FTA if
each of its $\varepsilon$-transitions is of the form
$q \xrightarrow{\top\slash0} q'$
and
each of its read-transitions is of the form
$q \xrightarrow{\top\slash0} \sigma(p_1, \dots, p_n)$.
We drop the counter values when we denote computations of FTA.

\subsection{Examples}

We consider a few examples for tree languages recognizable by GOCTA in order to show their capabilities.

\begin{example}\label{ex:a=b}
    Let $\Sigma=\{\sigma^{(3)},a^{(1)},b^{(1)},\#^{(0)}\}$ and consider the tree language 
    \[L_{a=b}=\{\xi=\sigma(\xi_1,\xi_2,\xi_3)\mid\xi_1,\xi_2,\xi_3\in T_{\{a,b,\#\}},|\xi|_{a}=|\xi|_{b}\}\,.\]
    This language can by recognized by a GOCTA as follows: Let $\A=(Q,\Sigma,q_0,\Delta)$ where $Q=\{q_0\}\cup\{[u]_v\mid u,v\in\{a,b\}\}$ and $\Delta$ consists of the following transitions:
    \begin{itemize}
        \item $q_0\xrightarrow{\top\slash 0} \sigma([a]_v,[v]_{v'},[v']_{v''})$ for each $v,v',v''\in\{a,b\}$,
        \item $[u]_v\xrightarrow{\top\slash 1}u([u]_v)$, $[u]_v\xrightarrow{{>}0\slash -1}u'([u]_v)$, and $[u]_v\xrightarrow{0\slash 1}u'([u']_v)$ for $u,u',v\in\{a,b\}$ with $u\neq u'$, 
        and
        \item $[u]_u\xrightarrow{\top\slash 0} \#$ for each $u\in\{a,b\}$.
    \end{itemize}
    As the counter always stays positive, $\A$ needs to encode in its state whether it currently counts $a$s or $b$s. It has to guess with its first transition in which counting mode the recognition of $\xi_1$ (and $\xi_2$) ends in order to start the recognition of the following subtree in that state. \exqed
\end{example}

We note that we can easily generalize Example \ref{ex:a=b} and construct a GOCTA which recognizes the set  $T_{a=b}$ of all trees $\xi$ over $\Sigma'\supseteq\{a,b,\#\}$ with $|\xi|_a=|\xi|_b$. However, the restricted form of $L_{a=b}$ will be useful to show in Section \ref{sec:express} that this language is not recognizable by a one-counter tree automaton.

Now let us consider two more examples which demonstrate the expressiveness of GOCTA by their ability to multiply the counter value by a constant.
In these examples, we will write trees that are combs over binary symbols, i.e.,
trees of the form $\underbrace{\sigma(\dots(\sigma(}_{\text{$k$ times}}\#, \#), \underbrace{\#), \dots \#)}_{\text{$k-1$ times}}$, simply as
$\sigma^k((\#, \#), \#^{k-1})$.

\begin{figure}[t]
    \begin{subfigure}{0.49\textwidth}
        \begin{align*}
            (\zeta[p,q'], c)
            &\Rightarrow^{\tau_2} (\zeta[\sigma(p, q), q'], c-1) \\
            &\Rightarrow^{\tau_2} (\zeta[\sigma(\sigma(p, q), q), q'], c-2) \\
            & \dots \\
            &\Rightarrow^{\tau_2}
            (\zeta[\sigma^c((p, q), q^{c-1}), q'], 0) \\
            &\Rightarrow^{\tau_3}
            (\zeta[\sigma^c((\#, q), q^{c-1}), q'], 0) \\
            &\Rightarrow^{\tau_1}
            (\zeta[\sigma^c((\#, \#), q^{c-1}), q'], k) \\
            &\Rightarrow^{\tau_1}
            (\zeta[\sigma^c((\#, \#), \#q^{c-2}), q'], k \cdot 2) \\
            &\dots \\
            &\Rightarrow^{\tau_1}
            (\zeta[\sigma^c((\#, \#), \#^{c-1}), q'], k \cdot c)
        \end{align*}
    \end{subfigure}
    \begin{subfigure}{0.5\textwidth}
        \includegraphics[width=\textwidth]{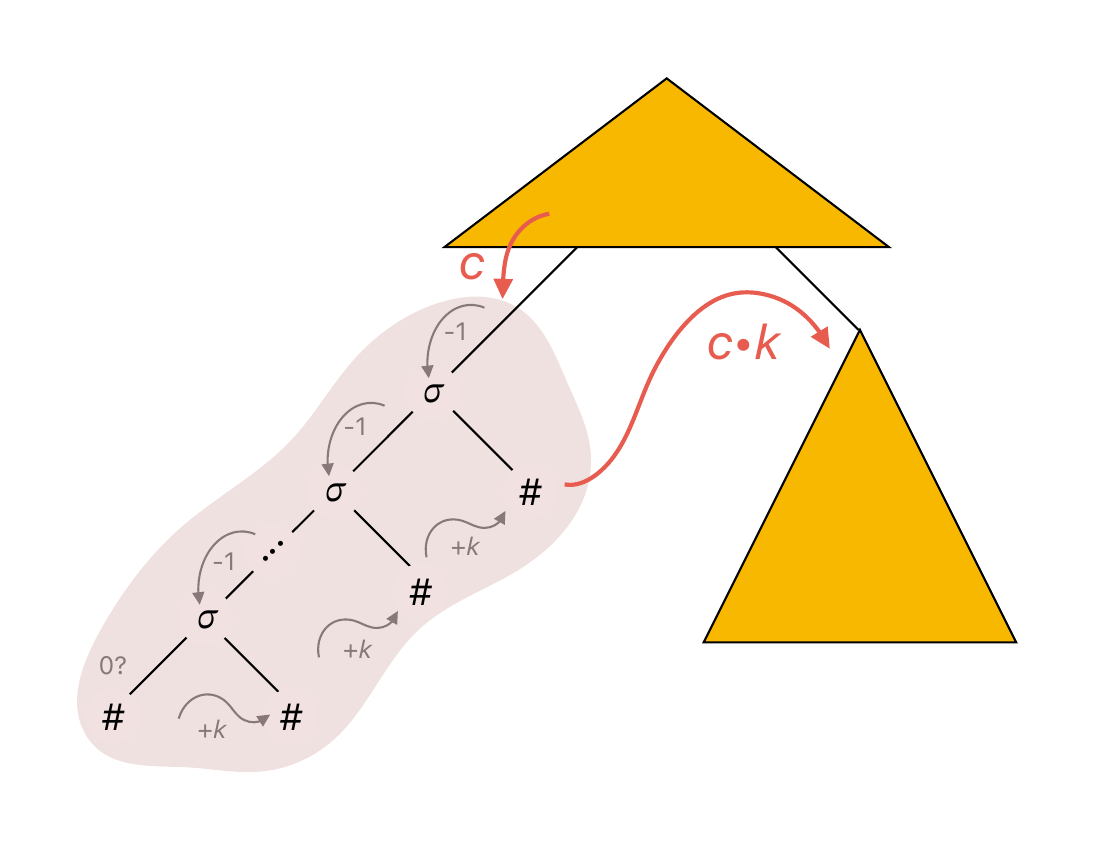}
    \end{subfigure}
    \caption{Left: computation of the GOCTA of Example~\ref{ex:multiply}
    which multiplies the counter value by $k$.
    Right: graphical representation of that computation.}
    \label{fig:multiply}
\end{figure}

\begin{example}\label{ex:multiply}
    Let $\Sigma$ be a ranked alphabet, $\sigma^{(2)}, \#^{(0)}\in\Sigma$,
    and $k \in \Nat$.
    We consider a GOCTA $\A$ whose set of transitions contains
    \[
        (\tau_1) \
        q \xrightarrow{\top \slash k} \#
        \qquad
        (\tau_2) \
        p \xrightarrow{{>}0 \slash -1} \sigma(p, q)
        \qquad
        (\tau_3) \
        p \xrightarrow{0 \slash 0} \#
        .
    \]
    Now let $\zeta \in C_\Sigma(X_2)$,
    $c \in \Nat$, and $q'$ be a state of $\A$.
    In Fig.\;\ref{fig:multiply}, we illustrate how
    $(\zeta[p,q'],c) \Rightarrow_{\A}^*
    (\zeta[\sigma^c((\#, \#), \#^{c-1}), q'], k \cdot c)$.
    Thus, recognizing the subtree
    $\sigma^c((\#, \#), \#^{c-1})$,
    the GOCTA $\A$ can multiply the counter value~by~$k$.\exqed
\end{example}

\begin{example}\label{ex:k^n}
    Now we extend Example \ref{ex:multiply} as follows: Let $\Sigma=\{\omega^{(2)},\sigma^{(2)}, \#^{(0)}\}$ and consider the GOCTA $\A=(\{q,p,f,q_0\},\Sigma,q_0,\Delta)$ where $\Delta$ consists of the transitions $\tau_1,\tau_2,\tau_3$ from Example \ref{ex:multiply} for $k = 2$ plus the transitions
    \[(\tau_0)\ q_0 \xrightarrow{\top \slash 1} \omega(p,f) \qquad (\tau_4)\ f\xrightarrow{\top \slash 0}\omega(p,f)\qquad (\tau_5)\ f\xrightarrow{\top \slash 0}\#.\]
    Then $(q_0,0)\Rightarrow^{\tau_0}(\omega(p,f),1)\Rightarrow^*(\omega(\zeta_1,f),2)\Rightarrow^{\tau_4}(\omega(\zeta_1,\omega(p,f)),2)$ and, by iterating these steps, $(\omega(\zeta_1,\omega(p,f)),2)\Rightarrow^*(\omega(\zeta_1,\omega(\zeta_2,\ldots,\omega(\zeta_n,\#)\ldots)),2^n)$ where $\zeta_i=\sigma^{(2^{i-1})}((\#, \#), \#^{(2^{i-1})-1})$.
    Intuitively, each computation of $\zeta_i$ doubles the counter value as outlined in Example~\ref{ex:multiply}.
    Thus, for each tree $\xi\in\L(\A)$ we obtain $|\xi|_\sigma=2^n-2$ for some $n\in\Nat$ and, since there is no bound for $|\xi|_\sigma$, this language is not semilinear.\exqed
\end{example}

\subsection{Normal Forms and Maximal Counter Values}

Now we introduce two normal forms of GOCTA which will simplify the proofs in the remaining work. Moreover, we will also show that, when analysing computations of GOCTA, we can restrict ourselves to those containing counter values polynomial in the size of the computed tree. This insight will later be a key ingredient in our proof that the membership problem of GOCTA is in \textsc{P}.

Let $\A$ be a GOCTA over $\Sigma$.
We set \[\L_0(\A)=\{\xi\in T_\Sigma\mid (q_0,0)\Rightarrow_\A^*(\xi,0)\}\]
and call $\A$ \emph{0-accepting} if $\L_0(\A)=\L(\A)$.
We say that $\A$ is \emph{normalized} if $\textsc{Instr}(\A)\subseteq\{-1,0,1\}$ and each of its read-transitions is of the form $q\xrightarrow{\top\slash 0} \sigma(p_1,\ldots, p_n)$.

\begin{lemma}\label{lem:0-accepting}
    Let $\A$ be a GOCTA. Then there is a 0-accepting GOCTA $\A'$ such that $\L(\A')=\L(\A)$.
\end{lemma}

\begin{proof} In order to be 0-accepting, the GOCTA $\A'$ we construct needs to empty its counter before reading the rightmost leaf symbol. To do so, the automaton must know whether it is currently operating in the rightmost part of a tree. 
We encode this information in the states: for each state $q$ of $\A$, we add an auxiliary state $[q]$, which will always be passed to (and only to) the rightmost child starting from the new initial state $[q_0]$. 


    Formally, let $\A=(Q,\Sigma,q_0,\Delta)$ be a GOCTA. We construct the GOCTA $\A'=(Q',\Sigma,[q_0],\Delta')$ as follows: Let $Q'=Q\cup\{[q]\mid q\in Q\}\cup\{q_\alpha\mid q\in Q,\alpha\in\Sigma\}$ and $\Delta'=\Delta\cup T$ where $T$ consists of the following transitions:
    \begin{itemize}
        \item for each $\varepsilon$-transition $q\xrightarrow{\p\slash \z} q'\in \Delta$, the transition $[q]\xrightarrow{\p\slash \z} [q']$ is in $T$,
        \item for each transition $q\xrightarrow{\p\slash \z} \sigma(q_1,\ldots,q_n)\in\Delta$, the transition\\ $[q]\xrightarrow{\p\slash \z} \sigma(q_1,\ldots,q_{n-1}, [q_n])$ is in $T$,
        \item for each transition $q\xrightarrow{\p\slash \z}\alpha\in\Delta$, the transition $[q]\xrightarrow{\p\slash \z}q_\alpha$ is in $T$,
        \item for each $q\in Q$ and $\alpha\in\Sigma$, the transition $q_\alpha\xrightarrow{{>}0\slash -1}q_\alpha$ is in $T$, and
        \item for each $q\in Q$ and $\alpha\in\Sigma$, the transition $q_\alpha\xrightarrow{0\slash 0} \alpha$ is in $T$.
    \end{itemize}
    Clearly, $\A'$ is 0-accepting and one can show that for each $\xi\in T_\Sigma$, $\xi\in \L(\A)$ if and only if $\xi\in\L(\A')$.\qed
\end{proof} 

\begin{observation}
    Let $\A$ be a GOCTA.
    Then there is a normalized GOCTA $\A'$ such that $\L(\A')=\L(\A)$.
    Moreover, if $\A$ is 0-accepting, then $\A'$ is 0-accepting.
\end{observation}

Now we show that we can restrict the counter values occurring in computations of GOCTA. We note that a similar proof idea was used in \cite{AmaJea13a} to obtain the pumping lemma for pushdown automata.

\begin{lemma}\label{lem:polybound}
    Let $L\subseteq T_\Sigma$. If $L$ is GOCTA-recognizable, then there is a normalized and $0$-accepting GOCTA $\A=(Q,\Sigma,q_0,\Delta)$ with $\L(\A)=L$ and such that for each $\xi\in L$ there is a computation $t\in\compA{\xi}$ with $\mathrm{maxcnt}(t)\leq |\xi|\cdot|Q|^2+1$.
\end{lemma}

\begin{proof}
    Let $\A=(Q,\Sigma,q_0,\Delta)$ be a normalized and 0-accepting GOCTA. Let $\xi\in L$ and $t\in\compA{\xi}$ such that $\mathrm{maxcnt}(t)> |\xi|\cdot|Q|^2+1$.
    The idea of the proof is that, for the recognition of $\xi$, only the counter values up to $|\xi|\cdot|Q|^2+1$ contain relevant information, whereas $\A$ can no longer distinguish counter values greater than $|\xi|\cdot|Q|^2+1$ from each other: both counting up into and counting down from that range must necessarily contain state cycles within which $\A$ cannot test for $0$. We have chosen $\xi$ and $t$ such that $\mathrm{maxcnt}(t)$ is large enough to find two such cycles, one counting up and one counting down, that a) cause the same change (in terms of absolute amount) of the counter and b) occur between the reading of two symbols. We can show that we can cut out both cycles and the resulting computation still recognizes $\xi$.
    
    We let $m=\mathrm{maxcnt}(t)$. By the assumptions, $t$ is of the form $(q_0,0)\Rightarrow^\theta(\ang{\xi}_i[p,\bar q_i],m)\Rightarrow^\omega(\xi,0)$ for some $i\in\{0,\ldots,|\xi|-1\}$ such that the counter value $m$ is reached after $\theta$ for the first time.
    As $m>|\xi|\cdot|Q|^2+1$ and $i<|\xi|$ as well as $\textsc{Instr}(\A')\subseteq\{-1,0,1\}$, there exist $j\in\{0,\ldots,i\}$, $m_j^1\geq 1,m_j^2=m_j^1+(|Q|^2+1)\leq m$ and $p_1,p_2\in Q$ such that $\theta$ is of the form $(q_0,0)\Rightarrow^{\theta_1}(\ang{\xi}_j[p_1,\bar q_j],m_j^1)\Rightarrow^{\theta_2}(\ang{\xi}_j[p_2,\bar q_j],m_j^2)\Rightarrow^{\theta_3}(\ang{\xi}_i[p,\bar q_i],m)$ and we choose the latest such occurrences. Note that $\theta_2\in\Delta_\varepsilon^*$.

    Similarly, there are $l\in\{i,\ldots,{|\xi|-1}\}$, $m_l^1\leq m,m_l^2=m_l^1-(|Q|^2+1)\geq 1$ and $f_1,f_2\in Q$ such that $\omega$ is of the form $(\ang{\xi}_i[p,\bar q_i],m)\Rightarrow^{\omega_1}(\ang{\xi}_l[f_1,\bar q_l],m_l^1)\Rightarrow^{\omega_2}(\ang{\xi}_l[f_2,\bar q_l],m_l^2)\Rightarrow^{\omega_3}(\xi,0)$ and now we choose the first such occurrences in $\omega$. By choosing $\theta_2$ and $\omega_2$ closest to $(\ang{\xi}_i[p,\bar q_i],m)$, we obtain:

    \begin{property}
        For all counter values $m'$ occurring in configurations of the subcomputation $(\ang{\xi}_j[p_2,\bar q_j],m_j^2)\Rightarrow^{\theta_3}(\ang{\xi}_i[p,\bar q_i],m)\Rightarrow^{\omega_1}(\ang{\xi}_l[f_1,\bar q_l],m_l^1)$ we have $m'>|Q|^2+1$.
    \end{property}

    Now we zoom into the part $t=(\ang{\xi}_j[p_1,\bar q_j],m_j^1)\Rightarrow^{\theta_2}(\ang{\xi}_j[p_2,\bar q_j],m_j^2)$: let, for $z\in\{0,\ldots,|Q|^2+1\}$, $p^z\in Q$ such that $(\ang{\xi}_j[p^z,\bar q_j],m_j^1+z)$ is the configuration in $t$ with $m_j^1+z$ occurring for the last time. Clearly, $p^{|Q|^2+1}=p_2$.
    Similarly, $f^z$ denotes the state with which $m_l^1-z$ occurs for the first time in $\omega_2$. We obtain:
    
    \begin{property}
        For each $z\in\{0,\ldots,|Q|^2+1\}$ and each counter value $m'$ (except the first one) occurring in the subcomputation $(\ang{\xi}_j[p^z,\bar q_j],m_j^1+z)\Rightarrow^*(\ang{\xi}_j[p_2,\bar q_j],m_j^2)$ of $\theta_2$ we have $m'>m_j^1+z$. Similarly, for each counter value $m''$ (except the last one) occurring in the subcomputation $(\ang{\xi}_l[f_1,\bar q_l],m_l^1)\Rightarrow^*(\ang{\xi}_l[p^z,\bar q_j],m_j^1-z)$ of $\omega_2$ we have $m''>m_l^1-z$.
    \end{property}
    
    Now consider $Z=\{(p^z,f^z)\mid z\in[|Q|^2+1]\}$. Clearly, $|Z|\leq|Q|^2$. Thus, there are $z_1< z_2\in[|Q|^2+1]$ with $p^{z_1}=p^{z_2}$ and $f^{z_1}=f^{z_2}$.

    But this means that we can find a shorter computation of $\A$ for $\xi$ by cutting out the subcomputations between $p^{z_1}$ and $p^{z_2}$ as well as $f^{z_1}$ and $f^{z_2}$ at the respective position indicated above. We note that, by the above Property 1, all transitions applied in $\theta_3$ and $\omega_1$ can only use as predicate ${>} 0$ or $\top$. Thus, in the following we mainly need to argue that also after our cut we never reach 0 in the counter during these sequences and, thus, all transitions are still applicable.  Let $d=z_2-z_1$ and $\theta_2=\theta_{z_1}\theta_{z_1,z_2}\theta_{z_2}$, $\omega_2=\omega_{z_1}\omega_{z_1,z_2}\omega_{z_2}$ where $\theta_{z_1,z_2}$ and $\omega_{z_1,z_2}$ are the parts we cut out, respectively. Clearly, $(q_0,0) \Rightarrow^{\theta_1\theta_{z_1}}(\ang{\xi}_j[p^{z_1},\bar q_j],m_j^1+z_1)$. By Property 2, $(\ang{\xi}_j[p^{z_1},\bar q_j],m_j^1+z_1)\Rightarrow^{\theta_{z_2}}(\ang{\xi}_j[p_2,\bar q_j],m_j^2-d)$ as all transitions are still applicable on counter values decreased by $d$. With a similar argumentation, by Property 1, $(\ang{\xi}_j[p_2,\bar q_j],m_j^2-d)\Rightarrow^{\theta_3\omega_1}(\ang{\xi}_l[f_1,\bar q_l],m_l^1-d)$ as $d\leq|Q|^2+1$.
    Finally, by Property 2, $(\ang{\xi}_l[f_1,\bar q_l],m_l^1-d)\Rightarrow^{\omega_{z_2}}(\ang{\xi}_l[f^{z_2},\bar q_l],m_l^1-{z_2})$ and $(\ang{\xi}_l[f^{z_2},\bar q_l],m_l^1-{z_2})\Rightarrow^*(\xi,0)$.

    This procedure can be repeated until no counter value greater than $|\xi|\cdot|Q|^2+1$ does occur anymore.\qed
\end{proof}

\section{Expressiveness}\label{sec:express}

We investigate the expressiveness of GOCTA in comparison to OCTA;
recall that, at each instance of branching,
OCTA pass a copy of the counter to every successor.
We formalize OCTA as an alternative semantics of GOCTA.

Let $\A=(Q,\Sigma,q_0,\Delta)$ be a GOCTA and denote by $\textsc{ID}$ the set $Q\times\Nat$. For each transition $\tau\in\Delta$ we let $\Rightarrow^\tau_\cp$ be the binary relation on the set $T_\Sigma(\textsc{ID})$ such that for each $\zeta_1,\zeta_2\in T_\Sigma(\textsc{ID})$ we have
\(\zeta_1\Rightarrow^\tau\zeta_2\)
if there are $\hat\zeta\in C_\Sigma(\textsc{ID},X_1)$, $\hat\zeta_1,\hat\zeta_2\in T_\Sigma(\textsc{ID})$ such that $\zeta_1=\hat\zeta[\hat\zeta_1]$, $\zeta_2=\hat\zeta[\hat\zeta_2]$, and either
\begin{itemize}
    \item $\tau=q\xrightarrow{\p\slash \z} q'$, $\hat\zeta_1=(q,m)$, $\p(m)$, $m+\z\geq 0$, and $\hat\zeta_2=(q',m+\z)$, or
    \item $\tau=q\xrightarrow{\p\slash \z} \sigma(p_1,\ldots, p_n)$, $\hat\zeta_1=(q,m)$, $\p(m)$, $m+\z\geq0$, and\\ $\hat\zeta_2=\sigma((q_1,m+\z),\ldots,(q_n,m+\z))$.
\end{itemize}
Thus, transitions can be applied at any position carrying a state.
We let $\Rightarrow_\cp=\bigcup_{\tau\in\Delta}\Rightarrow^\tau_\cp$ and $\L_{\cp}(\A)=\{\xi\in T_\Sigma\mid (q_0,0)\Rightarrow_{\cp}^*\xi\}$. A tree language $L$ \emph{can be recognized by an OCTA} if there exists a GOCTA $\A$ with $\L_{\cp}(\A)=L$.

\begin{example}
    Let $\Sigma=\{\sigma^{(2)},a^{(1)},b^{(1)},c^{(1)},\#^{(0)}\}$ and consider the tree language
    \(L_{a\sigma{}bc}=\{a^n(\sigma(b^n\#,c^n\#))\mid n\in\Nat\}\,.\)
    This language can be recognized by the OCTA $\A=(\{q_0,q_1,q_2\},\Sigma,q_0,\Delta)$ where $\Delta$ contains the transitions ${q_0\xrightarrow{\top\slash 1}a(q_0)}$, $q_0\xrightarrow{\top\slash 0}\sigma(q_1,q_2)$,
    $q_1\xrightarrow{{>} 0\slash -1}b(q_1)$, $q_2\xrightarrow{{>} 0\slash -1}c(q_2)$, $q_1\xrightarrow{0\slash 0}\#$, and $q_2\xrightarrow{0\slash 0}\#$.\exqed
\end{example}

\begin{lemma}\label{lem:lasbc-not-gocta}
    The tree language $L_{a\sigma{}bc}$ is not recognizable by a GOCTA.
\end{lemma}

\begin{proof}
    The proof is based on the following intuition.
    In order to check whether a tree $\xi$ is in $L_{a\sigma{}bc}$,
    a GOCTA $\A$ must traverse $\xi$ in lexicographical order.
    Thus, it effectively checks whether $h(\xi)$
    is in $L = \{ a^n \sigma b^n \# c^n \# \mid n \in \Nat \}$
    where $h$ is the tree homomorphism induced by
    $h(\sigma) = \sigma x_1 x_2$, $h(a) = a x_1$, $h(b) = b x_1$, $h(c) = c x_1$,
    and $h(\#) = \#$.
    We show that there exists a one-counter string automaton $\A'$
    such that $\A'$ recognizes $L$ if and only if $\A$ recognizes $L_{a\sigma{}bc}$.
    Since $L$ is not context-free, this entails that no GOCTA can recognize $L_{a\sigma{}bc}$.

    Formally, assume that there exists a normalized and $0$-accepting GOCTA
    $\A = (Q, \Sigma, q_0, \Delta)$ such that $L(\A) = L_{a\sigma{}bc}$.
    We observe that, for every $\xi \in L(\A)$ and $t \in \compA{\xi}$,
    the computation $t$ has the form
    \begin{equation}\label{eq:lgamma3-comp}
        t =(q_0,0)\Rightarrow^{ \theta_a \; (q_1 \to \sigma(q_2, q_3)) \; \theta_b \; (q_4 \to \#) \; \theta_c \; (q_5 \to \#)}(\xi,0)
    \end{equation}
    where $\theta_a, \theta_b, \theta_c \in \Delta^*$
    only contain $a$, $b$, and $c$, respectively, in their read transitions.
    Without loss of generality, we can assume that the sets of states occurring in
    $\theta_a$, $\theta_b$, and $\theta_c$ are pairwise disjoint
    (thus, e.g., $q_1$ can only occur in $\theta_a$).

    We only sketch the definition of the one-counter string automaton $\A'$;
    its set of transitions $\Delta'$ is obtained from $\Delta$ as follows.
    For each transition $q_1 \to \sigma(q_2, q_3)$,
    we add $q_1 \to \sigma\langle q_2, q_3\rangle$ to $\Delta'$
    where $\langle q_2, q_3\rangle$ becomes a state of $\A'$;
    then, for every transition $q_4 \to \#$ and for every transition of the form
    $q \to b(q')$ or $q \xrightarrow{\p \slash \z} q'$ that can occur in $\theta_b$,
    we add $\langle q_4,q_3\rangle \to \# q_3$, $\langle q, q_3\rangle \to b\langle q', q_3\rangle$,
    or $\langle q, q_3\rangle \xrightarrow{\p \slash \z} \langle q', q_3\rangle$ to $\Delta'$, resp.
    Thus, intuitively, we encode the single instance of branching in $\A$
    into the state behaviour of $\A'$.
    All other transitions of $\A$
    are added to $\Delta'$ unchanged.

    One can show that,
    for every $\xi \in L(\A)$ and $t \in \compA{\xi}$,
    there exists a matching $t' \in \compAp{h(\xi)}$.
    Vice versa,
    for every $w \in L(\A')$ and $t \in \compAp{w}$,
    there exists $\xi \in L(\A)$ and a matching $t' \in \compA{\xi}$
    such that $h(\xi) = w$.
    Thus, we obtain that $h(L_{a\sigma{}bc}) = L = L(\A')$.
    Since $L$ is not context-free, this is a contradiction.
    Thus, $\A'$ cannot exist and, hence, $\A$ does not recognize $L_{a\sigma{}bc}$.\qed
\end{proof}

\begin{lemma}\label{lem:lab-not-octa}
    $L_{a=b}$ cannot be recognized by OCTA.
\end{lemma}

\begin{proof}
    If there exists an OCTA $\A = (Q, \Sigma, q_0, \Delta)$ with $\L_\cp(\A) = L_{a=b}$,
    we can assume that there are disjoint $Q_1, Q_2 \subseteq Q$ such that $Q_1 \cup Q_2 = Q$ and each computation of $\A$ has the form $(q_0,0) \Rightarrow_{\A}^* (q,m) \Rightarrow_{\A} \sigma((q_1,m), (q_2,m), (q_3,m)) \Rightarrow_{\A}^* \sigma(\xi_1,\xi_2,\xi_3)$ with $q_1,q_2,q_3 \in Q_2$.
    Moreover, for each $q \in Q_2$, no state in $Q_1$ is reachable.
    For every $n \in \Int$, we let $L_n = \{ w\# \mid w \in \{a,b\}^*, |w|_a = |w|_b + n\}$,
    and, for every $c \in \Nat$ and $q \in Q$, we let $L_c^q = \{ \xi\in T_\Sigma \mid (q,c) \Rightarrow_{\A}^* \xi \}$.
    We obtain:

    \begin{property}\label{prop:lcq-subset-lm}
        For every $(q,c)\in\textsc{ID}$ with $q \in Q_2$ that occurs in a successful computation, there exists $n \in \Int$ such that $L_c^q \subseteq L_n$.
    \end{property}
    If Property~\ref{prop:lcq-subset-lm} did not hold,
    then $\A$ could derive $w_1, w_2 \in \{a,b\}^*\#$
    such that $|w_1|_a - |w_1|_b \not= |w_2|_a - |w_2|_b$ from $(q,c)$.
    Since $(q,c)$ occurs in a successful computation,
    there would exist a $\xi \in \L_\cp(\A)$ with $|\xi|_a \not= |\xi|_b$;
    hence $\L_\cp(\A) \not= L_{a=b}$.

    We observe that, for every $\xi \in L_{a=b}$,
    there exist unique $n_1, n_2, n_3 \in \Int$ such that
    $\xi \in \sigma(L_{n_1}, L_{n_2}, L_{n_3})$.
    We let $k \in \Nat$ and consider the set
    \begin{align*}
        F(k) &= \{ (n_1, n_2, n_3) \mid \exists \xi \in L_{a=b}\colon \xi \in \sigma(L_{n_1}, L_{n_2}, L_{n_3}), |n_1|, |n_2|, |n_3| \le k \} \\
        &\overset{\star}= \{ (n_1,n_2,n_3) \mid n_1,n_2,n_3 \in \{ -k, \dots, +k \}, n_1 + n_2 + n_3 = 0 \}
    \end{align*}
    where $\star$ holds by definition of $L_{a=b}$.
    Note that, if $\height(\xi) \le k+1$, then $n_1, n_2, n_3 \le k$, and,
    for every $(n_1, n_2, n_3) \in F(k)$,
    there is a $\xi \in L_{a=b} \cap \sigma(L_{n_1}, L_{n_2}, L_{n_3})$
    with $\height(\xi) \le k+1$.
    Using combinatorial arguments, one can show $|F(k)| = 3k^2 + 3k + 1$.

    Let $\xi, \xi' \in L_{a=b}$ and $n_1, n_2, n_3, n_1', n_2', n_3' \in \Int$ such that
    $\xi \in \sigma(L_{n_1}, L_{n_2}, L_{n_3})$
    and $\xi' \in \sigma(L_{n_1'}, L_{n_2'}, L_{n_3'})$.
    By definition of $\A$, there exist
    ${q_1, q_2, q_3, q_1', q_2', q_3' \in Q_2}$
    and
    $c, c' \in \Nat$
    such that
    $\xi \in \sigma(L_c^{q_1}, L_c^{q_2}, L_c^{q_3})$ and $\xi' \in \sigma(L_{c'}^{q_1'}, L_{c'}^{q_2'}, L_{c'}^{q_3'})$.
    Now, by Property~\ref{prop:lcq-subset-lm}, we obtain the following.
    If $(n_1, n_2, n_3) \not= (n_1', n_2', n_3')$,
    then also $(c, q_1, q_2, q_3) \not= (c', q_1', q_2', q_3')$.
    We define $D(k)$ to be the set
    \[
        \{ (c, q_1, q_2, q_3) \! \mid \! (q_0,0) \Rightarrow_{\A}^* (q,c) \Rightarrow_{\A} \sigma((q_1,c), (q_2,c), (q_3,c)) \Rightarrow_{\A}^* \xi, \height(\xi) \le k+1 \}.
    \]
    Thus, $|D(k)|$ must be at least $|F(k)|$, i.e., $|D(k)| \ge 3k^2 + 3k + 1$.
    Then there exists $(c, q_1, q_2, q_3) \in D(k)$ with $c \ge \frac{3k^2 + 3k + 1}{|Q|^3}$.
    Using a method like in Lemma~\ref{lem:polybound},
    we can show that, if $\xi\in\L_\cp(\A)$,
    then there exist $m\leq\height(\xi)\cdot |Q|^{4}+1$ and
    a computation of the form
    \((q_0,0)\Rightarrow^*(q,m)\Rightarrow\sigma((q_1,m),(q_2,m),(q_3,m))\Rightarrow^*\xi\).
    Thus, for each $(c, q_1, q_2, q_3) \in D(k)$,
    we can assume $c \le (k+1) \cdot |Q|^4 + 1$.
    This is a contradiction for sufficiently large $k$.\qed
\end{proof}

\begin{corollary}\label{cor:incomparable}
    The classes of tree languages recognizable by OCTA and GOCTA are incomparable.
\end{corollary}


\section{Decision problems}

Now we turn to decidability results. In the following, we examine the emptiness problem and the membership problem of GOCTA in this regard.

\subsection{Emptiness}
\label{sec:emptiness}

First we investigate the question whether the emptiness problem is decidable for GOCTA. Unfortunately, the answer is negative -- we can relate GOCTA to a similar string formalism: \emph{indexed counter grammars with a global counter semantics} (ICG) \cite[Sec. 5]{DusMidPar92} (called \emph{R-mode derivation}). This model can be understood as a context-free grammar with a global counter controlling its derivations and was shown to be Turing-complete \cite[Thm. 5.1]{DusMidPar92}. We note that, in contrast, with GOCTA we cannot perform the same calculations without reading symbols. However, we can show that, given an ICG $G$, we can construct a GOCTA $\A$ recognizing all computation trees of $G$. Thus, we obtain that the emptiness problem of ICG can be reduced to the emptiness problem of GOCTA.

\begin{theorem}\label{thm:emptiness}
    Given a GOCTA $\A$, it is undecidable whether $\L(\A)=\emptyset$.
\end{theorem}

An intuition on this result is given by Example \ref{ex:multiply}: In \cite[Thm. 5.1]{DusMidPar92} it was shown that with help of multiplication, a counter can be used to simulate a pushdown storage by representing each of its configurations as a number with base $m$ (where $m$ is the size of the pushdown alphabet). Since a context-free grammar using a global pushdown is Turing-complete, this entails the result.

\subsubsection{Relation to GIG.}
The above undecidability result seems to be in contrast with~\cite{Castano04},
where it was shown that the languages generated by Global Index Grammars (GIG)
are semilinear, from which it was deduced that the emptiness problem of GIG is decidable.
In essence, a GIG is a context-free grammar enhanced with a global pushdown storage
where, additionally, each pushing rule
must start its right-hand side with a terminal.

We briefly recall the definition of GIG
\cite[Definition~6]{Castano04}.
A GIG is a tuple $G = (N, T, I, S, \bot, P)$ where
$N$, $T$, and $I$ are pairwise disjoint alphabets
(\emph{nonterminals}, \emph{terminals}, and \emph{stack indices}, respectively),
$S \in N$ (\emph{start symbol}),
$\bot \not\in N \cup T \cup I$ (\emph{stack start symbol}),
and $P$ is a finite set (\emph{productions}),
each of which having one of the following forms:
\[
    (i.1) \
    A \xrightarrow[\varepsilon]{} \alpha
    \qquad
    (i.2) \
    A \xrightarrow[{[y]}]{} \alpha
    \qquad
    (ii.) \
    A \xrightarrow[x]{} a \alpha
    \qquad
    (iii.) \
    A \xrightarrow[\bar x]{} \alpha
\]
where
$x \in I$, $y \in I \cup \{\bot\}$, $A \in N$, $\alpha \in (N \cup T)^*$, and $a \in T$.
Intuitively, $(i.1)$ is a context-free production,
$(i.2)$ checks if the topmost pushdown symbol is $y$ (without changing the pushdown),
$(ii.)$ pushes $x$, and $(iii.)$ pops $x$.

By assuming GIG with only one pushdown symbol (plus the bottom marker),
we can relate this formalism to GOCTA.
For instance, we can define a GIG $\A'$ with $I = \{\gamma,\bot\}$ which generates strings encoding the trees of the GOCTA $\A$ from Example~\ref{ex:k^n}. Therefore, we let $\A'$ contain the transitions
\begin{itemize}
    \item $q_0 \xrightarrow[\gamma]{} \omega(p,f)$ for $q_0 \xrightarrow{\top \slash 1} \omega(p,f)$,
    \item $p \xrightarrow[\bar \gamma]{} \sigma(p, q)$ for $p \xrightarrow{{>}0 \slash -1} \sigma(p, q)$,
    \item $p \xrightarrow[{[\bot]}]{} \#$ for $p \xrightarrow{0 \slash 0} \#$, and
    \item $f \xrightarrow[\varepsilon]{} \omega(p,f)$ and $f \xrightarrow[\varepsilon]{} \#$ for $f\xrightarrow{\top \slash 0}\omega(p,f)$ and $f\xrightarrow{\top \slash 0}\#$, respectively,
\end{itemize}
by regarding the symbols of $\A$ (including \enquote{(}, \enquote{)}, and \enquote{,}) as string terminals. Moreover, for the transition $\tau_1 = (q \xrightarrow{\top \slash 2} \#)$, which uses a $2$-increment, we add two rules to $\A'$ that each push a $\gamma$. Formally, we add the symbol $\beta^{(1)}$ to $\Sigma$,
use the new state $q_1$,
and add the rules $q \xrightarrow[\gamma]{} \beta(q_1)$ and $q_1 \xrightarrow[\gamma]{} \#$ to $\A'$.
Finally, we note that in the original definition, derivations of GIG are required to end with
the pushdown consisting of $\bot$.
This corresponds to $0$-accepting GOCTA and, clearly,
changing $\A'$ accordingly
does not violate the constraints of GIG.

Since $L(\A)$ is not semilinear, it is not hard to see that also $L(\A')$ is not semilinear.
This falsifies the claim that GIG generate only semilinear languages~\cite[Chap. 4, Thm. 5]{Castano04}
and also casts doubt on the emptiness result.

\begin{conj}
    The emptiness problem of GIG~\cite{Castano04} is undecidable.
\end{conj}

\subsection{Membership}

Surprisingly, in contrast to emptiness, the (uniform) membership problem for GOCTA is not only decidable but also turns out to be rather easy: we will show with help of the polynomial bound on the counter values in computations we obtained in Lemma \ref{lem:polybound}, that we can decide membership in polynomial time.

\begin{definition}
    Let $\A = (Q, \Sigma, q_0, \Delta)$ be a normalized GOCTA and $k \in \Nat$.
    The \emph{$k$-bounded behaviour automaton of $\A$} is the FTA
    $\A' = (Q', \Sigma, q_0', \Delta')$ where
    $Q' = Q \times [0,k] \times [0,k]$,
    $q_0' = (q_0, 0, 0)$, and $\Delta'$ is defined as follows.
    For every $n \geq 1$, $q \to \sigma(q_1,\ldots, q_n)$ in $\Delta$,
    and $i_0, \dots, i_n \in [0,k]$,
    we let $(q,i_0,i_n) \to \sigma((q_1,i_0,i_1), \ldots, (q_n,i_{n-1},i_n))$
    in $\Delta'$,
    for every $q \to \alpha$ in $\Delta$ and $i \in [0,k]$,
    we let $(q,i,i) \to \alpha$ in $\Delta'$, and,
    for every $q\xrightarrow{\p\slash \z} q'$ in $\Delta$ and $i_0, i_1 \in [0,k]$,
    if $\p(i_0)$ holds and $0 \le i_0 + \z \le k$, then we let
    $(q,i_0,i_1) \to (q',i_0+\z,i_1)$ in~$\Delta$.
\end{definition}

\begin{observation}\label{obs:behaviour-automaton-size}
    $|Q'| \le |Q| \cdot (k+1) \cdot (k+1)$ and
    $|\Delta'| \le |\Delta| \cdot (k+1)^{(\max\rk(\Sigma)+1)}$.
\end{observation}

\begin{lemma}\label{lem:behaviour-automaton-language}
    Let $\A = (Q, \Sigma, q_0, \Delta)$ be a GOCTA and $\xi \in T_\Sigma$.
    If $\A'$ is the $(|\xi|\cdot|Q|^2+1)$-bounded behaviour automaton of $\A$, then
    $\xi \in \L(\A) \iff \xi \in \L(\A')$.
\end{lemma}

\begin{proof}
    We assume that $\A$ is normalized and $0$-accepting. Clearly, these constructions are polynomial wrt. the size of $\A$.
    Let $m = |\xi|\cdot|Q|^2+1$ and $\A' = (Q', \Sigma, (q_0, 0, 0), \Delta')$.
    In order to show that $\xi \in \L(\A) \implies \xi \in \L(\A')$,
    we let $\xi \in L(\A)$.
    Then, by Lemma~\ref{lem:polybound},
    there exists $t \in \compA{\xi}$
    such that $\maxcnt(t) \le m$.
    We let $t = (\zeta_0, c_0) \Rightarrow^{\tau_1 \cdots \tau_n} (\zeta_n, c_n)$
    and we note that $\zeta_0 = q_0$, $\zeta_n = \xi$, and $c_0 = c_n = 0$.
    The idea of the proof is to construct $t' \in \compAp{\xi}$,
    denoted by $t' = \zeta_0' \Rightarrow^{\tau'_1 \cdots \tau'_n} \zeta_n'$
    such that, for every $i \in \{0,\dots,n\}$, the following invariant holds.
    If $\zeta_i = \ang{\xi}_k[q, \bar{q}]$,
    then $\zeta_i' = \ang{\xi}_k[(q,c_i,c'), \bar{p}]$ with $|\bar{q}| = |\bar{p}|$.
    We note that most of this construction is straightforward and only highlight
    two interesting cases.
    Let $i \in [n]$, $\zeta_{i-1} = \ang{\xi}_k[q, \bar{q}]$, and $\zeta_{i-1}' = \ang{\xi}_k[(q,c_{i-1},c'), \bar{p}]$.

    If $\tau_i = q \to \sigma(p_1, \dots, p_\ell)$ with $\ell > 1$,
    then, for each $j \in [2,\ell]$, we encode the counter value
    of $\A$ when it derives $p_j$ into
    the $j$-th state of~$\tau_i'$.
    For this, let $u_j \in [n]$ be the step
    when $p_j$ is derived in the computation of $\xi$ in~$\A$.
    Then
    \begin{equation}\label{eq:branching-form}
        \tau'_i = (q, c_{i-1}, c') \to \sigma((p_1, c_{i-1}, c_{u_2}), (p_2, c_{u_2}, c_{u_3}), \dots, (p_\ell, c_{u_\ell}, c'))
        \ \in \ \Delta'.
    \end{equation}
    Clearly, $\zeta_{i-1}' \Rightarrow^{\tau'_i} \zeta'_i = \ang{\xi}_{k+1}[(p_1, c_{i-1}, c_{u_2}), (p_2, c_{u_2}, c_{u_3}), \dots, (p_\ell, c_{u_\ell}, c'), \bar{p}]$.

    If $\tau_i = q \to \alpha$ and $|\bar{p}| \not= 0$,
    then there exists $i' \in [i-1]$ such that $\tau_{i'}$ is of form~\eqref{eq:branching-form}
    and the current state $(q, c_{i-1}, c')$ is derived from some state $(p,c'_1,c'_2)$ in the right-hand side of $\tau_{i'}$.
    We choose the largest such $i'$.
    By definition of $\Delta'$, the derivation of $(q, c_{i-1}, c')$ from $(p,c'_1,c'_2)$
    does not change the third component, hence $c' = c'_2$.
    Moreover, the state which occurs right of $(p,c'_1,c'_2)$
    in $\tau_{i'}$ has the form $(p',c'_2,c'_3)$
    and this is also the first state in $\bar{p}$.
    Then, $c'_2 = c_i$ (cf.\@ the previous case) and $c_{i-1} = c_i$ because $\A$ is normalized, hence
    $\tau'_i = (q,c_{i-1},c') \to \alpha$ is in $\Delta'$.
    Clearly, $\zeta_{i-1}' \Rightarrow^{\tau'_i} \zeta'_i = \ang{\xi}_{k+1}[\bar{p}]$.


    In order to show that $\xi \in \L(\A') \implies \xi \in \L(\A)$,
    we let $\xi \in \L(\A')$.
    Then there exists $t \in \compAp{\xi}$;
    we let $t = \zeta_0 \Rightarrow^{\tau_1 \cdots \tau_n} \zeta_n$
    and we note that $\zeta_0 = (q_0, 0, 0)$ and $\zeta_n = \xi$.
    Now we construct
    $t' = (\zeta_0', c_0) \Rightarrow^{\tau'_1 \cdots \tau'_n} (\zeta_n', c_n)$
    in $\compA{\xi}$
    analogously to the other direction of the proof,
    but now we have to show that the state behaviour of $\A'$
    guarantees a valid counter behaviour of $\A$.
    \qed
\end{proof}

Due to the following well-known result,
we can construct an $\varepsilon$-transition free FTA $\A''$ such that $\L(\A') = \L(\A'')$
in polynomial time.

\begin{lemma}[cf., e.g., {\cite[Theorem\,1.1.5]{tata}}]\label{lem:fta-epsilon-free}
    For every FTA $\A$ there exists an $\varepsilon$-transition free FTA $\A'$ with $\L(\A') = \L(\A)$
    and $\A'$ can be constructed in polynomial time wrt.\@ the size of $\A$.
\end{lemma}

\begin{theorem}\label{thm:membership}
    For every GOCTA $\A$ over $\Sigma$ and $\xi \in T_\Sigma$
    it can be decided in polynomial time whether $\xi \in \L(\A)$.
\end{theorem}

\begin{proof}
    Let $Q$ be the set of states of $\A$.
    We let $m = |\xi|\cdot|Q|^2+1$ and
    we construct the $m$-bounded behaviour automaton of $\A$ and call it $\A'$.
    By Observation~\ref{obs:behaviour-automaton-size},
    the size of $\A'$ is polynomial in $m$ and the size of $\A$.
    Clearly, the construction of $\A'$ has the same time bound.
    Now, by Lemma~\ref{lem:fta-epsilon-free},
    there exists an $\varepsilon$-transition free FTA $\A''$ with $\L(\A'') = \L(\A')$.
    Moreover, $\A''$ can be computed in polynomial time wrt.\@ the size of $\A'$.
    We can compute whether $\xi \in \L(\A'')$
    in polynomial time wrt.\@ the size of $\A''$ and $\xi$
    (cf., e.g., \cite[Theorem 1.7.3]{tata}),
    and, thus, in polynomial time wrt.\@ $m$.
    Finally, it holds that 
    \begin{equation*}
        \xi \in \L(\A) \overset{\text{L.\,\ref{lem:behaviour-automaton-language}}}\iff
        \xi \in \L(\A') \overset{\text{\cite[Th.\,1.1.5]{tata}}}\iff \xi \in \L(\A''). \tag*{\qed}
    \end{equation*}
\end{proof}

\section{Conclusion}

We have introduced global one-counter tree automata (GOCTA)
and shown that they extend the capabilities of one-counter tree automata (OCTA)
in certain ways, but are less expressive in other regards.
We have shown that the membership problem of GOCTA is in P
while, in contrast to expectations raised by literature on global index grammars,
their emptiness problem is undecidable.

Future work should investigate closure properties
of the class of tree languages recognizable by GOCTA.
While closure under Boolean operations will likely follow the string case,
i.e., closure under union holds but closure under intersection does not,
closure under different restrictions of (inverse) tree homomorphisms is an interesting topic for further investigations.

\subsubsection*{Acknowledgments.} This work was supported by the European Research Council through the ERC Consolidator Grant No. 771779 (DeciGUT).

We thank Johannes Osterholzer and Sebastian Rudolph
for their valuable ideas which
have pointed us into the right direction to obtain several of the results in this work.

\printbibliography

\appendix

\section{Global One-Counter Tree Automata (Section 3)}

Let $\A$ be an OCTA. We say that $\A$ is \emph{0-accepting} if, for each symbol $\alpha\in\Sigma^{(0)}$, each of $\A$s read-transitions for $\alpha$ is of the from $q\xrightarrow{0\slash 0}\alpha$ for some state $q\in Q$, i.e., the counter has to be emptied before $\A$ reads a leaf node. Moreover, we call $\A$ \emph{normalized} if $\textsc{Instr}(\A)\subseteq\{-1,0,1\}$ and each of its read-transitions (except those one reading leaf nodes) is of the form $q\xrightarrow{\top\slash 0} \sigma(p_1,\ldots, p_n)$.

\begin{observation}
    Given an OCTA $\A$ there is a 0-accepting and normalized OCTA $\A'$ such that $\L(\A)=\L(\A')$.
\end{observation}

\begin{lemma}\label{lem:octa-bound}
    Let $\Sigma$ be a ranked alphabet such that $\Sigma^{(3)}=\{\sigma\}$ and $\Sigma^{(i)}=\emptyset$ for each $i\in\Nat\setminus\{0,1,3\}$. Let $T_{\Sigma,\sigma}$ be the set of all trees over $\Sigma$ where $\sigma$ occurs only at the root and let $L\subseteq T_{\Sigma,\sigma}$. If $L$ is OCTA-recognizable then there is an OCTA $\A$ with $\L(\A)=L$ such that each $\xi=\sigma(\xi_1,\xi_2,\xi_3)\in L$ is recognized by a computation of the form 
    \[(q_0,0)\Rightarrow^*(q,m)\Rightarrow\sigma((q_1,m),(q_2,m),(q_3,m))\Rightarrow^*\xi\]
    where $m\leq\height(\xi)\cdot |Q|^{4}+1$ with $Q$ being the set of states of $\A$.
\end{lemma}

\begin{proof}
    The proof is very similar to the full proof of Lemma \ref{lem:polybound} and, thus, presented with less details. Let $\A=(Q,\Sigma,q_0,\Delta)$ be an OCTA recognizing $L$. We may assume w.l.o.g. that $\A$ is 0-accepting and normalized.
    Let $\xi=\sigma(\xi_1,\xi_2,\xi_3)\in L$ and assume a computation $\theta$ of the form
    \begin{align*}
       (q_0,0) \Rightarrow^*(q,m)\Rightarrow\sigma((q_1,m),(q_2,m),(q_3,m))\Rightarrow^*\xi
    \end{align*} 
    of $\A$ where $m>\height(\xi)\cdot |Q|^{4}+1$. For each $i\in[3]$ assume $\xi_i=\gamma_{i,1}\ldots\gamma_{i,l_i}(\alpha_i)$ for some $l_i\in\Nat$ and symbols $\gamma_{i,1},\ldots,\gamma_{i,l_i},\alpha_i\in\Sigma$.  Consider the subcomputation $\theta_i=((q_i,m)\Rightarrow^*\xi_i)$. As $\A$ is 0-accepting, $\theta_i$ is of the form
    \begin{align*}
        (q_i,m)&\Rightarrow^*(q_i^1,m_i^1)\Rightarrow\gamma_{i,1}[(\hat{q}_i^1,m_i^1)]\\
        &\Rightarrow^*\gamma_{i,1}[(q_i^2,m_i^2)]\Rightarrow\gamma_{i,1}\gamma_{i,2}[(\hat{q}_i^2,m_i^2)]\\
        &\qquad \vdots\\
        &\Rightarrow^*\gamma_{i,1}\ldots\gamma_{i,l_i}[(q_i^{l_i+1},0)]\Rightarrow\gamma_{i,1}\ldots\gamma_{i,l_i}(\alpha_i).
    \end{align*}
    As $m>\height(\xi_i)\cdot|Q|^4+1$, there is a $u_i\in[l_i]$ with $m_i^{u_i}-m_i^{u_i+1}>|Q|^4$; we choose the smallest such occurrence. 
    As $\textsc{Instr}(\A)\subseteq\{-1,0,1\}$, there are $m_i > |Q|^4$, $m_i'=m_i-(|Q|^4+1)$ and $f_i^1,f_i^2\in Q$ such that there exists a subcomputation of $\theta_i$ of the form \[\hat\theta_i\colon t_i[(\hat{q}_i^{u_i},m_i^{ u_i})]\Rightarrow^*t_i[(f_i^1,m_i)]\Rightarrow^*t_i[(f_i^2,m_i')]\Rightarrow^*t_i[(q_i^{u_i+1},m_i^{u_i+1})]\]
    where $t_i=\gamma_{i,1}\ldots\gamma_{i,u_i}$ and we choose the first such occurrences.

    With the same argumentation as above, there are $m_0,m_0'\leq m$ with $m_0'=m_0-(|Q|^4+1)$ and $f_0^1,f_0^2\in Q$ such that 
    \[(q_0,0) \Rightarrow^*(f_0^1,m_0')\Rightarrow^*(f_0^2,m_0)\Rightarrow^*(q,m)\]
    and we choose the last such occurence.

    In combination, we obtain
    \begin{align*}
        (q_0,0) &\Rightarrow^*(f_0^1,m_0')\Rightarrow^*(f_0^2,m_0)\Rightarrow^*(q,m)\\
                    &\Rightarrow\sigma((q_1,m),(q_2,m),(q_3m))\\
                    &\Rightarrow^*\sigma(t_1[(f_1^1,m_1)],t_2[(f_2^1,m_2)],t_3[(f_3^1,m_3)])\\
                    &\Rightarrow^*\sigma(t_1[(f_1^2,m_1')],t_2[(f_2^2,m_2')],t_3[(f_3^2,m_3')])\\
                    &\Rightarrow^*\xi.
    \end{align*}

    Now let, for each $j\in[|Q|^4+1]$, $p_0^j\in Q$ such that $(f_0^1,m_0')\Rightarrow^*(p_0^j,m_0'+j)\Rightarrow^*(f_0^2,m_0)$ in $\theta$ and $m_0'+j$ occurs for the last time in this subcomputation. Similarly, let $p_i^j\in Q$ such that $t_i[(f_i^1,m_i)]\Rightarrow^*t_i[(p_i^j,m_i-j)]\Rightarrow^*t_i[(f_i^2,m_i')]$ and $m_i-j$ occurs for the first time in this subcomputation.

    Consider $Z=\{(p_0^j,p_1^j,p_2^j,p_3^j)\mid j\in[|Q|^4+1]\}$. Clearly, $|Z|\leq |Q|^4$. Thus, there are $j_1< j_2\in[|Q|^4+1]$ with $p_i^{j_1}=p_i^{j_2}$ for each $i\in\{0,1,2,3\}$.

    But this means that we can find a shorter computation of $\A$ for $\xi$ by cutting out the subcomputations between $p_i^{j_1}$ and $p_i^{j_2}$ at the respective position indicated above. Note that by selecting cutting points as close as possible to the chosen occurrence of $(q,m)$, we ensure that the counter values in the subcomputations between the cuts always stay above $|Q|^4 + 1$ and, thus, all transitions that were applicable before the cuts, are still applicable afterwards. For more details we refer to the proof of Lemma \ref{lem:polybound}. Now let $d=j_2-j_1$. Clearly, $(q_0,0) \Rightarrow^*(f_0^1,m_0')\Rightarrow^*(p_0^{j_1},m_0'+j_1)\Rightarrow^*(f_0^2,m_0-d)\Rightarrow^*(q,m-d)$. Moreover, we also obtain that $(q,m-d)\Rightarrow\sigma((q_1,m-d),(q_2,m-d),(q_3,m-d))$  and $(q_i,m-d)\Rightarrow^*t_i[(f_i^1,m_i-d)]\Rightarrow^*t_i[(p_i^{j_2},m_i-d)]$. Finally, as $(p_i^{j_2},m_i-d)\Rightarrow^*(f_i^2,m_i')$, we obtain $(q_0,0)\Rightarrow^*(q,m-d)\Rightarrow\sigma((q_1,m-d),(q_2,m-d),(q_3,m-d))\Rightarrow^*\xi$.\qed
\end{proof}

\section{Expressiveness (Section 4)}

We define \emph{one-counter (string) automata} (OCA)
as a special case of GOCTA where
a GOCTA $\A = (Q, \Sigma, q_0, \Delta)$ is an OCA
if $\max \rk(\Sigma) = 1$.
Since, then, the right-hand sides of the transitions of $\A$ are monadic trees,
we use our convention to denote them as strings.
It is well-known that the class of languages recognizable by OCA
is a subclass of the class of context-free languages.

\medskip

\noindent
\textbf{Lemma\;\ref{lem:lasbc-not-gocta}.}
\textit{The tree language $L_{a\sigma{}bc}$ is not recognizable by a GOCTA.}

\begin{proof}
    Assume that there exists a normalized and $0$-accepting GOCTA
    $\A = (Q, \Sigma, q_0, \Delta)$
    such that $L(\A) = L_{a\sigma{}bc}$.
    We observe that, for every $\xi \in L(\A)$ and $t \in \compA{\xi}$ of the form $(q_0,0) \Rightarrow^{\theta} (\xi,0)$,
    there exist $\theta_a, \theta_b, \theta_c \in \Delta^*$,
    $q_1, q_2, q_3, q_4, q_5 \in Q$ such that 
    \begin{equation}\label{eq:lgamma3-comp}
        \theta = \theta_a \; (q_1 \to \sigma(q_2, q_3)) \; \theta_b \; (q_4 \to \#) \; \theta_c \; (q_5 \to \#)
    \end{equation}
    and, for each $v \in \{a,b,c\}$ and read-transition $\tau$ occurring in $\theta_v$,
    we have that $\tau$ is of the form $q \to v(q')$.
    Thus, again without loss of generality, we assume that
    there exist pairwise disjoint $Q_a, Q_b, Q_c \subseteq Q$
    such that $Q = Q_a \cup Q_b \cup Q_c$ and,
    for every $\xi \in L_{a\sigma{}bc}$ and $\theta \in \compA{\xi}$, the following holds.
    If $\theta$ has form $\eqref{eq:lgamma3-comp}$, then,
    for each $v \in \{a,b,c\}$,
    every transition occurring in $\theta_v$ contains only states from $Q_v$
    (thus, in particular, $q_0, q_1 \in Q_a$, $q_2, q_4 \in Q_b$, and $q_3, q_5 \in Q_c$).
    We assume that every transition in $\Delta$ occurs in a computation of some $\xi \in L(\A)$,
    hence $\Delta$ contains no other kinds of read-transitions.

    We define the OCA $\A' = (Q', \Sigma', q_0', \Delta')$
    and the mapping $\varphi\colon \Delta' \to \Delta$ as follows.
    \begin{itemize}
        \item $Q' = Q_a \cup (Q_b \times Q_c) \cup Q_c$.
        \item $\Sigma' = \Sigma$ (ignoring the rank of each symbol).
        \item $q_0' = q_0$.
        \item For every $\tau = (q_1 \to \sigma(q_2, q_3))$ in $\Delta$,
            we let $\tau' = (q_1 \to \sigma  \langle q_2,q_3 \rangle $
            where $ \langle q_2, q_3 \rangle  \in Q_b \times Q_c$;
            then $\tau' \in \Delta'$ and $\varphi(\tau') = \tau$.
        \item For every $\tau = (q \to \#)$ in $\Delta$,
            if $q \in Q_c$,
            we let $\tau' = (q \to \#)$, $\tau' \in \Delta'$,
            and $\varphi(\tau') = \tau$;
            otherwise $q \in Q_b$, then, for every $q' \in Q_c$,
            we let $\tau_{q'} = ( \langle q,q' \rangle  \to \# q')$, $\tau_{q'} \in \Delta'$,
            and $\varphi(\tau_{q'}) = \tau$.
        \item For every $v \in \{a,b,c\}$ and $\tau = (q \to v(q'))$ in $\Delta$,
            if $q \in Q_a \cup Q_c$,
            we let $\tau' = (q \to v q')$, $\tau' \in \Delta'$,
            and $\varphi(\tau') = \tau$;
            otherwise $q \in Q_b$, then, for every $q'' \in Q_c$,
            we let $\tau_{q''} = ( \langle q,q'' \rangle  \to v  \langle q',q'' \rangle )$, $\tau_{q''} \in \Delta'$,
            and $\varphi(\tau_{q''}) = \tau$.
        \item Every other $\tau \in \Delta$ is an $\varepsilon$-transition,
            i.e., $\tau$ is of the form $q \xrightarrow{\p \slash \z} q'$.
            Now, if $q \in Q_a \cup Q_c$,
            we let $\tau' = (q \xrightarrow{\p \slash \z} q')$, $\tau' \in \Delta'$,
            and $\varphi(\tau') = \tau$;
            otherwise $q \in Q_b$, then, for every $q'' \in Q_c$,
            we let $\tau_{q''} = ( \langle q,q'' \rangle  \xrightarrow{\p \slash \z}  \langle q',q'' \rangle )$,
            $\tau_{q''} \in \Delta'$, and $\varphi(\tau_{q''}) = \tau$.
    \end{itemize}
    We lift $\varphi$ to a mapping $\varphi\colon \Delta'^* \to \Delta^*$
    in the obvious way.
    Moreover, we define the mapping $h\colon \Sigma \to (\Sigma \cup X_k)^*$
    by $h(\#) = \#$, $h(\sigma) = \sigma x_1 x_2$, and, for each $v \in \{a,b,c\}$,
    $h(v) = v x_1$;
    we lift $h$ to trees in the obvious way.

    We show that, for every $w \in L(\A')$ and $t' \in \compAp{w}$ of the form $(q_0',0) \Rightarrow^{\theta'} (w,0)$,
    there exists $\xi \in L(\A)$ and $t \in \compA{\xi}$
    such that $t$ is of the form $(q_0,0) \Rightarrow^{\varphi(\theta')} (\xi,0)$ and $h(\xi) = w$.
    Moreover, for every $\xi \in L(\A)$ and $t \in \compA{\xi}$,
    there exists $t' \in \compAp{h(\xi)}$ of the form $(q_0',0) \Rightarrow^{\theta'} (\xi,0)$
    such that $\varphi(\theta') = \theta$.

    Let $w \in L(\A')$ and $t' \in \compAp{w}$ of the form $(q_0',0) \Rightarrow^{\theta'} (w,0)$.
    By definition of $\A'$,
    there exist $\theta'_a, \theta'_b, \theta'_c \in \Delta'^*$ and
    $q_1, q_2, q_3, q_4, q_5 \in Q$ such that 
    \begin{equation}\label{eq:lgamma3string-comp}
        \theta' = \theta'_a \; (q_1 \to \sigma  \langle q_2, q_3\rangle) \; \theta'_b \; ( \langle q_4,q_3 \rangle  \to \# q_3 ) \; \theta'_c \; (q_5 \to \#),
    \end{equation}
    for each $v \in \{a,c\}$, every read-transition occurring in $\theta'_v$
    is of the form $q \to v q'$ with $q, q' \in Q_v$
    and
    every $\varepsilon$-transition occurring in $\theta'_v$ contains only states from $Q_v$,
    every read-transition occurring in $\theta'_b$
    is of the form $ \langle q,q_3 \rangle  \to b  \langle q',q_3 \rangle $ with $q, q' \in Q_b$,
    and
    every $\varepsilon$-transition occurring in $\theta'_b$ contains only states
    of the form $ \langle q,q_3 \rangle $ with $q \in Q_b$.
    Let $\theta$ denote $\varphi(\theta')$;
    clearly, $\theta \in \Delta^*$ and $\theta$ is of the form~\eqref{eq:lgamma3-comp}.

    Since $(q_0,0) \Rightarrow^{\theta'} w$,
    one can verify by induction that there exists $\xi \in L(\A)$ with
    $(q_0,0) \Rightarrow^{\theta} \xi$ and $h(\xi) = w$.
    First, we remark that, for every transition $\tau$ occurring in $\theta'$,
    the counter behaviour of $\tau$ is identical to that of $\varphi(\tau)$.
    Moreover, if $\tau$ occurs in $\theta'_a$, $\theta'_c$, or
    it is the last transition $q_5 \to \#$,
    then also their state behaviour is identical.
    Otherwise, we note that exactly one transition of the form
    $\tau_\sigma = q_1 \to \sigma  \langle q_2, q_3 \rangle $ occurs in $\theta'$.
    Hence there exist $c \in \Nat$ and $\zeta \in C_\Sigma(X_2)$ such that
    $(q_0,0) \Rightarrow^{\varphi(\theta'_a \tau_\sigma)} (\zeta[q_2,q_3], c)$.
    Then, $\theta'_b$ keeps the second component of the state tuple as $q_3$
    while $\varphi(\theta'_b)$, since it is a left-most computation,
    changes only the first state, too, i.e.,
    there exist $c' \in \Nat$ and $\zeta' \in C_\Sigma(X_1)$ such that
    $(q_0,0) \Rightarrow^{\varphi(\theta'_a \tau_\sigma \theta'_b ( \langle q_4,q_3 \rangle  \to \# q_3))} (\zeta'[q_3], c')$.
    Thus, and since $(\zeta'[q_3], c') \Rightarrow^{\varphi(\theta'_c) (q_5 \to \#)} \xi$, there exists $\xi \in L(\A)$ such that $(q_0,0) \Rightarrow^{\theta} \xi$.
    Moreover, we note that every transition $\tau$ occurring in $\theta$
    has the same symbol from $\Sigma$ as $\varphi(\tau)$.
    Since $(q_0,0) \Rightarrow^{\theta} \xi$ is a left-most computation
    and $h$ essentially reads off the symbols from $\xi$ in lexicographical order,
    we obtain $w = h(\xi)$.

    Now we let $\xi \in L(\A)$ and $t \in \compA{\xi}$.
    We recall that $t = (q_0,0) \Rightarrow^\theta (\xi,0)$ with $\theta$ as in~\eqref{eq:lgamma3-comp}.
    We define $\theta' \in \Delta'^{|\theta|}$ such that, for each $i \in [|\theta'|]$
    and $q, q' \in Q$,
    \begin{itemize}
        \item if $\theta[i] = (q_1 \to \sigma(q_2, q_3))$,
            we let $\theta'[i] = (q_1 \to \sigma  \langle q_2,q_3 \rangle )$,
        \item if $\theta[i] = (q_5 \to \#)$,
            we let $\theta'[i] = (q_5 \to \#)$,
        \item if $\theta[i] = (q_4 \to \#)$,
            we let $\theta'[i] = ( \langle q_4,q_3 \rangle  \to \# q_3)$
        \item if $\theta[i] = (q \to a(q'))$,
            we let $\theta'[i] = (q \to a q')$,
        \item if $\theta[i] = (q \to b(q'))$,
            we let $\theta'[i] = ( \langle q,q_3 \rangle  \to b  \langle q',q_3 \rangle )$,
        \item if $\theta[i] = (q \to c(q'))$,
            we let $\theta'[i] = (q \to c q')$,
        \item if $\theta[i] = q \xrightarrow{\p \slash \z} q'$
            and $q \in Q_a \cup Q_c$,
            we let $\theta'[i] = (q \xrightarrow{\p \slash \z} q')$, and
        \item if $\theta[i] = q \xrightarrow{\p \slash \z} q'$
            and $q \in Q_b$,
            we let $\theta'[i] = ( \langle q,q_3 \rangle  \xrightarrow{\p \slash \z}  \langle q',q_3 \rangle )$.
    \end{itemize}
    Then $\theta'$ has form \eqref{eq:lgamma3string-comp} and $\varphi(\theta') = \theta$.
    One can verify by induction that $(q_0,0) \Rightarrow^{\theta'} h(\xi)$
    where the argumentation is essentially the same as in the other direction.

    To finish the proof, we note that $h(L_{a\sigma{}bc}) = L$.
    We show that $L = L(\A')$.
    First, let $w \in L$.
    Then there exists $\xi \in L_{a\sigma{}bc}$ such that $w = h(\xi)$.
    Since $L(\A) = L_{a\sigma{}bc}$, there exists $(q_0,0) \Rightarrow^\theta (\xi,0)$ in $\compA{\xi}$.
    We have shown that there exists $(q_0',0) \Rightarrow{\theta'} (w,0)$ is in $\compAp{w}$;
    thus, $w \in L(\A')$.
    Now let $w \in L(\A')$.
    Then there exists $(q_0',0) \Rightarrow^{\theta'} (w,0)$ in $\compAp{w}$.
    We have shown that
    there exists $\xi \in L(\A)$ such that $(q_0,0) \Rightarrow^{\varphi(\theta')} (\xi,0)$ is in $\compA{\xi}$
    and $w = h(\xi)$.
    Then, since $L(\A) = L_{a\sigma{}bc}$,
    we have $\xi \in L_{a\sigma{}bc}$ and hence $w = h(\xi) \in L$.

    We have shown that the OCA $\A'$ recognizes $L$.
    However, since $L$ is not context-free, this is a contradiction.
    Therefore, $\A$ does not recognize $L_{a\sigma{}bc}$.\qed
\end{proof}

\subsubsection{Lemma\;\ref{lem:lab-not-octa}.}
$L_{a=b}$ cannot be recognized by OCTA.

\begin{proof}
    Assume that there exists an OCTA $\A = (Q, \Sigma, q_0, \Delta)$ with $L(\A) = L_{a=b}$.
    We can assume that there exist $Q_1, Q_2 \subseteq Q$ such that $Q_1 \cup Q_2 = Q$, $Q_1 \cap Q_2 = \emptyset$ and each computation of $\A$ has the form $(q_0,0) \Rightarrow_{\A}^* (q,m) \Rightarrow_{\A} \sigma((q_1,m), (q_2,m), (q_3,m)) \Rightarrow_{\A}^* \sigma(\xi_1,\xi_2,\xi_3)$ with $q_1,q_2,q_3 \in Q_2$.
    Moreover, for each $q \in Q_2$, no state in $Q_1$ is reachable.
    For every $n \in \Int$, we let $L_n = \{ w(\#) \mid w \in \{a,b\}^*, |w|_a = |w|_b + n\}$.
    Moreover, for every $c \in \Nat$ and $q \in Q$, we let $L_c^q = \{ \xi \mid (q,c) \Rightarrow_{\A}^* \xi \}$.

    \sepstars

    \noindent\textit{Property\;\ref{prop:lcq-subset-lm}.}
    For every $(q,c) \in \textsc{ID}$ with $q \in Q_2$ that occurs in a successful computation, there exists $n \in \Int$ such that $L_c^q \subseteq L_n$.

    In order to see that this property holds, let $(q,c) \in \textsc{ID}$ with $q \in Q_2$ such that $(q,c)$ occurs in a successful computation and,
    for every $m \in \Nat$, it holds that $L_c^q \not\subseteq L_m$.
    Then there exist $m_1, m_2 \in \Int$ with $m_1 \not= m_2$ and $w_1, w_2 \in L_c^q$ such that $|w_1|_a - |w_1|_b = m_1$ and $|w_2|_a - |w_2|_b = m_2$.
    Since $(q,c)$ occurs in a successful computation, there exists $\zeta \in C_\Sigma(X_1)$ such that $(q_0,0) \Rightarrow_{\A}^* \zeta[(q,c)] \Rightarrow_{\A}^* \zeta[w_1]$ and $(q_0,0) \Rightarrow_{\A}^* \zeta[(q,c)] \Rightarrow_{\A}^* \zeta[w_2]$.
    Clearly, $|\zeta[w_1]|_a - |\zeta[w_1]|_b \not= |\zeta[w_2]|_a - |\zeta[w_2]|_b$.
    Hence, $\zeta[w_1] \not\in L_{a=b}$ or $\zeta[w_2] \not\in L_{a=b}$ and, thus, $L(\A) \not= L_{a=b}$.

    \sepstars

    We observe that, for every $\xi \in L_{a=b}$, there exist unique $m_1, m_2, m_3 \in \Int$ such that $\xi \in \sigma(L_{m_1}, L_{m_2}, L_{m_3})$.
    Let $n \in \Nat$.
    Clearly, if $\height(\xi) \le n+1$, then $m_1, m_2, m_3 \le n$.
    We let $k \in \Nat$ and consider the set $F(k) = \{ (m_1, m_2, m_3) \mid \exists \xi \in L_{a=b}\colon \xi \in \sigma(L_{m_1}, L_{m_2}, L_{m_3}), |m_1|, |m_2|, |m_3| \le k \}$.
    We note that, for every $(m_1, m_2, m_3) \in F(k)$, there exists $\xi \in L_{a=b} \cap \sigma(L_{m_1}, L_{m_2}, L_{m_3})$ with $\height(\xi) \le k+1$.
    By definition of $L_{a=b}$ it holds that $F(k) = \{ (m_1,m_2,m_3) \mid m_1,m_2,m_3 \in \{ -k, \dots, +k \}, m_1 + m_2 + m_3 = 0 \}$.

    \sepstars

    \begin{property}\label{prop:f(k)-cardinality}
        $|F(k)| = 3k^2 + 3k + 1$
    \end{property}

    Let $k \in \Nat$ and $(m_1,m_2,m_3) \in F(k)$.
    As $m_1 \in \{-k, \dots, +k\}$, there are $2k + 1$ choices for $m_1$.
    We fix $m_1 = \ell_1$.
    If $\ell_1 \le 0$, then, for every $\ell_2 \in \{ -k - \ell_1, \dots, k \}$, there exists exactly one $\ell_3 \in \{ -k, \dots, k \}$ such that $\ell_1 + \ell_2 + \ell_3 = 0$ and, for every other $\ell_2$, there exists no such $\ell_3$.
    If $\ell_1 > 0$, then, for every $\ell_2 \in \{ -k, \dots, k - \ell_1 \}$, there exists exactly one $\ell_3 \in \{ -k, \dots, k \}$ such that $\ell_1 + \ell_2 + \ell_3 = 0$ and, for every other $\ell_2$, there exists no such $\ell_3$.
    Thus, in each case, there are $k + 1 - |\ell_1|$ choices for $m_2$ and there is one choice for $m_3$.
    In total, there are $\sum_{\ell_1 = -k}^k 2k + 1 - |\ell_1|$ choices for $(m_1,m_2,m_3)$.
    We obtain
    \begin{align*}
        |F(k)| &= \sum_{\ell_1 = -k}^k 2k + 1 - |\ell_1| \\
        &= (2k+1) \cdot (2k+1) + \sum_{\ell_1=-k}^k |-\ell_1| \\
        &= (2k+1) \cdot (2k+1) -2 \cdot \sum_{\ell_1=1}^k \ell_1 \\
        &= (2k+1) \cdot (2k+1) -2 \cdot \frac{k \cdot (k+1)}{2} \\
        &= 4k^2 + 4k + 1 - (k^2 + k) \\
        &= 3k^2 + 3k + 1.
    \end{align*}

    \sepstars
    
    By definition of $\A$, there exist $c \in \Nat$ and $q_1, q_2, q_3 \in Q_2$ such that $\xi \in \sigma(L_c^{q_1}, L_c^{q_2}, L_c^{q_3})$.
    Let $\xi, \xi' \in L_{a=b}$, $m_1, m_2, m_3, m_1', m_2', m_3' \in \Int$, $c, c' \in \Nat$, and $q_1, q_2, q_3, q_1', q_2', q_3' \in Q_2$ such that
    $\xi \in \sigma(L_{m_1}, L_{m_2}, L_{m_3})$, $\xi' \in \sigma(L_{m_1'}, L_{m_2'}, L_{m_3'})$ as well as
    $\xi \in \sigma(L_c^{q_1}, L_c^{q_2}, L_c^{q_3})$ and $\xi' \in \sigma(L_{c'}^{q_1'}, L_{c'}^{q_2'}, L_{c'}^{q_3'})$.
    By Property~\ref{prop:lcq-subset-lm}, the following holds:
    if $(m_1, m_2, m_3) \not= (m_1', m_2', m_3')$, then $(c, q_1, q_2, q_3) \not= (c', q_1', q_2', q_3')$.
    We let
    \[
        D(k) = \{ (c, q_1, q_2, q_3) \mid \begin{aligned}[t]
            &(q_0,0) \Rightarrow_{\A}^* (q,c) \Rightarrow_{\A} \sigma((q_1,c), (q_2,c), (q_3,c)) \Rightarrow_{\A}^* \xi, \\
            &\height(\xi) \le k+1 \};
        \end{aligned}
    \]
    thus, $|D(k)|$ must be at least $|F(k)|$.
    Hence, by Property~\ref{prop:f(k)-cardinality}, we obtain $|D(k)| \ge 3k^2 + 3k + 1$.
    Thus, there exists $(c, q_1, q_2, q_3) \in D(k)$ with $c \ge \frac{3k^2 + 3k + 1}{|Q|^3}$.
    However, by Lemma~\ref{lem:octa-bound}, each $\xi \in L(\A)$ can be recognized by a computation with maximal counter value $\height(\xi) \cdot |Q|^4 + 1$.
    Thus, for every $(c, q_1, q_2, q_3) \in D(k)$, we can assume $c \le (k+1) \cdot |Q|^4 + 1$.
    This is a contradiction for sufficiently large $k$.\qed
\end{proof}

\section{Decision Problems (Section 5)}

\subsection{Emptiness}

We recall the definition of an indexed counter grammar from \cite{DusMidPar92}. For the counter, we use two symbols not occuring as nonterminal or terminal symbols in the grammar: $\#$ as a bottom marker and $\gamma$ as the single counter symbol; we set $I=\{\gamma,\#\}$. 

An \emph{indexed counter grammar} is a tuple $G=(N\cup\{S\},T,P,S)$ where $N$ and $T$ are disjoint sets of nonterminal and terminal symbols, resp., $S \not \in N \cup T$ is the initial nonterminal symbol, and $P$ is a finite set of productions of the forms
\begin{itemize}
    \item $S\to A\#$ and $S\to\varepsilon$,\hfill$(S1)$ and $(S2)$
    \item $Ag\to u_1 B_1\beta_1\ldots u_n B_n\beta_n u_{n+1}$, or\hfill$(T1)$
    \item $A\#\to u_1 B_1\beta_1\# u_2 B_2\beta_2 \ldots u_n B_n\beta_n u_{n+1}$,\hfill$(T2)$
\end{itemize}
where $n\in\Nat$, $A\in N$, $g\in\{\gamma,\varepsilon\}$, $u_1,\ldots,u_{n+1}\in T^*$, and $B_i\in N$, $\beta_i\in I^*$ for each $i\in[n]$. Note that in productions of form $(T2)$, the bottom marker $\#$ does only occur once\footnote{This definition for the form of $T2$-productions fixes a small error in \cite{DusMidPar92} where $R$-mode derivations where given as an alternative to the usual semantics which copyies the pushdown and distributes it to each nonterminal in the right-hand side of a production. For the copying semantics, it was necessary for $\#$ to occur after each $\beta_i$.}.

The semantics of $G$ (via $R$-mode derivation) is defined as follows. Given $\Theta=u_1B_1\beta_1 u_2 B_2\beta_2 \ldots u_n B_n\beta_n u_{n+1}$ with $u_1,\ldots,u_{n+1}\in T^*$, $B_i\in N$, $\beta_i\in I^*$, and $\delta\in I^*$, we define
\[\Theta :_R\delta=u_1B_1\beta_1\delta u_2 B_2\beta_2 \ldots u_n B_n\beta_n u_{n+1}.\] Now let $\Theta',\Theta''\in(N I^*\cup T)^*$ and $\tau\in P$. Then we let $\Theta'\Rightarrow_R^\tau\Theta''$ if
\begin{itemize}
    \item $\tau=S\to A\#$, $\Theta'=S$ and $\Theta''=A\#$, or $\tau=S\to\varepsilon$, $\Theta'=S$, and $\Theta''=\varepsilon$, or
    \item $\tau=Ag\to \Theta$ of form $(T1)$, $\Theta'=wAg\omega\Theta_1$ with $w\in T^*$, $\omega\in\gamma^*\#$, and $\Theta_1\in(N I^*\cup T)^*$, and $\Theta''=w(\Theta\Theta_1:_R\omega)$, or
    \item $\tau=A\#\to \Theta$ of form $(T2)$, $\Theta'=wA\#\Theta_1$ with $w\in T^*$ and $\Theta_1\in(N I^*\cup T)^*$, and $\Theta''=w\Theta\Theta_1$.
\end{itemize}
We set $L_R(G)=\{w\in T^*\mid S\Rightarrow_R^*w\}$.

A derivation is a sequence $t=\Theta_0\Rightarrow^{\tau_1}\Theta_1\ldots\Rightarrow^{\tau_n}\Theta_n$ (sometimes abbreviated as $\Theta_0\Rightarrow^{\tau_1\ldots\tau_n}\Theta_n$) such that $n\in\Nat$, $\Theta_0,\ldots,\Theta_n\in(N I^*\cup T)^*$, $\tau_1,\ldots,\tau_n\in P$, and $\Theta_{i-1}\Rightarrow^{\tau_i}\Theta_i$ for each $i\in[n]$. We call $t$ \emph{successful on $w\in T^*$} if $\Theta_0=S$ and $\Theta_n=w$; the set of all successful computations of $G$ on $w$ is denoted by $\derivG{w}$. We set $\derivGall=\bigcup_{w\in T^*}\derivG{w}$. Clearly, $L_R(G)=\{w\in T^*\mid\derivG{w}\neq\emptyset\}$.


Next we show that the non-emptiness problem of indexed counter grammars can be reduced to the non-emptiness problem of GOCTA. For this, we construct, given an indexed counter grammar $G$, a GOCTA $\A$ that recognizes all \enquote{computation trees} of $G$.

\begin{lemma}
    For each indexed counter grammar $G$ there is a GOCTA $\A$ such that $L(\A)\neq\emptyset$ if and only if $L_R(G)\neq\emptyset$.
\end{lemma}

\begin{proof}
    Let $G=(N\cup\{S\},T,P,S)$ be an indexed counter grammar. For each $\tau\in P$, we denote by $\mathrm{rk}(\tau)$ the number of nonterminal symbols occurring in its right-hand side. Let $\Sigma$ be the ranked alphabet defined by $\Sigma^{(n)}=\{\tau\in P\mid\mathrm{rk}(\tau)=n\}$ for each $n\in\Nat$. 

    Now in a first step let us define an indexed counter grammar $G'=(N\cup\{S\},\emptyset,P',S)$ where $P'$ is constructed from $P$ by replacing each symbol occurrence from $T$ in the right-hand side of a transition by $\varepsilon$. Clearly, $\derivGp{\varepsilon}\neq\emptyset$ if and only if $\derivGall\neq\emptyset$. Moreover, for this particular grammar we define the set $\derivGp{m_1,m_2}$ of derivations computing $\varepsilon$ starting from the counter state $\gamma^{m_1}\#$ and ending in the counter state $\gamma^{m_2}\#$ (for $m_1$, $m_2\in\Nat$): a derivation $\Theta_0\Rightarrow^{\tau_1\ldots\tau_n}\Theta_n$ is in $\derivGp{m_1,m_2}$ if and only if 
    \begin{itemize}
        \item $\Theta_0= A\gamma^{m_1}\#\Theta'$,
        \item $\Theta_{n-1}=B\gamma^{m_2'}\#$ with $m_2'=m_2+1$ if $\tau_n=B\gamma\to\varepsilon$ and $m_2'=m_2$ otherwise, and
        \item $\Theta_n=\varepsilon$
    \end{itemize}
    for some $A,B\in N$ and $\Theta'\in (NI^*)^*$. 

We construct the GOCTA $\A=(Q,\Sigma,S,\Delta)$ where $Q=N\cup\{S\}\cup\{[B\omega]\mid Af\to\Theta_1 B\omega\Theta_2\in P, f\in\{\#,\gamma,\varepsilon\}, \Theta_1,\Theta_2\in (NI^*)^*\}$ and $\Delta$ consists of the following transitions:
\begin{itemize}
    \item If $\tau=S\to\varepsilon$ is in $P$, then $S\xrightarrow{\top\slash 0} \tau$ is in $\Delta$.
    \item If $\tau=S\to A\#$ is in $P$, then $S\xrightarrow{\top\slash 0}\tau(A)$ is in $\Delta$.
    \item If $\tau=Af\to B_1\beta_1\ldots B_n\beta_n $ is in $P$, then $A\xrightarrow{\p\slash \z}\tau([B_1\beta_1],\ldots,[B_n\beta_n])$ is in $\Delta$ where
    \begin{itemize}
        \item $\p=\top$ and $\z=0$ if $f=\varepsilon$,
        \item $\p={>} 0$ and $\z=-1$ if $f=\gamma$, and
        \item $\p=0$ and $\z=0$ if $f=\#$.
    \end{itemize}
    \item For each $[B\gamma^n\#]$ in $Q$ we let $[B\gamma^n\#]\xrightarrow{\top\slash n}B$ in $\Delta$.
    \item For each $[B\gamma^n]$ in $Q$ we let $[B\gamma^n]\xrightarrow{\top\slash n}B$ in $\Delta$.
\end{itemize}

Define $h\colon T_\Sigma\to P^*$ by $h(\tau(\xi_1,\ldots \xi_n))=\tau h(\xi_1)\ldots h(\xi_n)$. We will show that $\xi\in\L(\A)$ if and only if $(S\Rightarrow^{h(\xi)}\varepsilon)\in\derivGp{\varepsilon}$. As $L_R(G')\neq\emptyset$ iff $\derivGp{\varepsilon}\neq\emptyset$, we obtain $L_R(G')\neq\emptyset$ iff $\L(\A)\neq\emptyset$.

\sepstars

\begin{property}\label{prop:empt1}
    Let $k\in\Nat$, $\xi\in T_\Sigma$, $m_1,m_2\in\Nat$, $A\in N$, and $\theta\in\Delta^k$ such that $(A,m_1)\Rightarrow^\theta(\xi,m_2)$. Then $(A\gamma^{m_1}\#\Rightarrow^{h(\xi)}\varepsilon)\in\derivGp{m_1,m_2}$.
\end{property}

We prove Property \ref{prop:empt1} by strong induction on $k$. As there is no such computation of length $0$, the induction basis trivially holds.

Now let $k>0$ and assume that Property \ref{prop:empt1} holds for all $k'\in\Nat$ with $k'<k$. Let $\theta=\bar\tau\theta'$ with $\bar\tau=A\xrightarrow{\p\slash \z}\tau([B_1\beta_1],\ldots,[B_n\beta_n])$. Then $\xi=\tau(\xi_1,\ldots,\xi_n)$, $\theta'=\theta_1\ldots\theta_n$, and there are $s_0,\ldots,s_{n}\in\Nat$ such that $s_0=m_1+\z$, $s_n=m_2$, and $([B_i,\beta_i],s_{i-1})\Rightarrow^{\theta_i}(\xi_i,s_i)$ for each $i\in[n]$. By construction, $\theta_i=\tau_i\theta_i'$ with $\tau_i=[B_i\beta_i]\xrightarrow{\top\slash k_i}B$ if $\beta_i\in\{\gamma^{k_i},\gamma^{k_i}\#\}$. Thus, $([B_i,\beta_i],s_{i-1})\Rightarrow^{\tau_i}(B,s_{i-1}+k_i)\Rightarrow^{\theta_i'}(\xi_i,s_i)$. By induction hypothesis, $(B\gamma^{s_{i-1}+k_i}\#\Rightarrow^{h(\xi_i)}\varepsilon)\in\derivGp{s_{i-1}+k_i,s_i}$. Now we proceed with a case distinction on $\bar\tau$:

\emph{Case 1:} Let $\bar\tau=A\xrightarrow{\top\slash 0}\tau([B_1\beta_1],\ldots,[B_n\beta_n])$. Then, by construction, $\tau=A\to B_1\beta_1\ldots B_n\beta_n$ is a production in $P$ and $\beta_i=\gamma^{k_i}$ for each $i \in [n]$. As $\textsc{Instr}(\bar\tau)=0$, we have $s_0=m_1$. But then we obtain
\begin{align*}
    A\gamma^{m_1}\#&\Rightarrow^\tau(B_1\gamma^{k_1}\ldots B_n\gamma^{k_n}):_R(\gamma^{s_0}\#)=B_1\gamma^{s_0+k_1}\#\ldots B_n\gamma^{k_n}\\
    &\Rightarrow^{h(\xi_1)}\varepsilon(B_2\gamma^{k_2}\ldots B_n\gamma^{k_n}):_R(\gamma^{s_1}\#)=B_2\gamma^{s_1+k_2}\#\ldots B_n\gamma^{k_n}\\
    &\quad\vdots\\
    &\Rightarrow^{h(\xi_n)}\varepsilon
\end{align*}
which is a derivation in $\derivGp{m_1,m_2}$.

\emph{Case 2:} Let $\bar\tau=A\xrightarrow{{>} 0\slash -1}\tau([B_1\beta_1],\ldots,[B_n\beta_n])$. Then $m_1>0$ and $s_0=m_1-1\geq0$. By construction, $\tau=A\gamma\to B_1\beta_1\ldots B_n\beta_n$ is a production in $P$ and $\beta_i=\gamma^{k_i}$ for each $i \in [n]$. We obtain
\begin{align*}
    A\gamma^{m_1}\#&\Rightarrow^\tau(B_1\gamma^{k_1}\ldots B_n\gamma^{k_n}):_R(\gamma^{s_0}\#)=B_1\gamma^{s_0+k_1}\#\ldots B_n\gamma^{k_n}\\
    &\Rightarrow^{h(\xi_1)}\varepsilon(B_2\gamma^{k_2}\ldots B_n\gamma^{k_n}):_R(\gamma^{s_1}\#)=B_2\gamma^{s_1+k_2}\#\ldots B_n\gamma^{k_n}\\
    &\quad\vdots\\
    &\Rightarrow^{h(\xi_n)}\varepsilon
\end{align*}
which is a derivation in $\derivGp{m_1,m_2}$.

\emph{Case 3:} Let $\bar\tau=A\xrightarrow{0\slash 0}\tau([B_1\beta_1],\ldots,[B_n\beta_n])$. Then $m_1=s_0=0$. By construction, $\tau=A\#\to B_1\beta_1\#\ldots B_n\beta_n$ is a production in $P$ and $\beta_i=\gamma^{k_i}$ for each $i \in [n]$. We obtain
\begin{align*}
    A\#&\Rightarrow^\tau B_1\gamma^{k_1}\#\ldots B_n\gamma^{k_n}\\
    &\Rightarrow^{h(\xi_1)}\varepsilon(B_2\gamma^{k_2}\ldots B_n\gamma^{k_n}):_R(\gamma^{s_1}\#)=B_2\gamma^{s_1+k_2}\#\ldots B_n\gamma^{k_n}\\
    &\quad\vdots\\
    &\Rightarrow^{h(\xi_n)}\varepsilon
\end{align*}
which is a derivation in $\derivGp{0,m_2}$.

\sepstars

\begin{property}\label{prop:empt2}
    Let $k\in\Nat$, $m_1,m_2\in\Nat$, $A\in N$, and $t\in P^k$ such that $(A\gamma^{m_1}\#\Rightarrow^{t}\varepsilon)\in\derivGp{m_1,m_2}$. Then there are $\xi\in h^{-1}(t)$ and $\theta\in\Delta^*$ with $(A,m_1)\Rightarrow^{\theta}(\xi,m_2)$.
\end{property}

We prove Property \ref{prop:empt2} by strong induction on $k$. As there is no such derivation of length $0$, the induction basis trivially holds.

Now let $k>0$ and assume that Property \ref{prop:empt2} holds for all $k'\in\Nat$ with $k'<k$. Let $t=\tau t'$ with $\tau=Af\to B_1\beta_1\ldots B_n\beta_n $ and choose $k_1,\ldots,k_n\in\Nat$ such that $\beta_i\in\{\gamma^{k_i},\gamma^{k_i}\#\}$ for each $i\in[k]$. As $(A\gamma^{m_1}\#\Rightarrow^{\tau t'}\varepsilon)\in\derivGp{m_1,m_2}$, $t'=t_1\ldots t_n$ and there are $s_0,\ldots,s_{n}\in\Nat$ such that $s_0=m_1-1$ if $f=\gamma$ and $s_0=m_1$ otherwise, $s_n=m_2$, and $(B\gamma^{s_{i-1}+k_i}\#\Rightarrow^{t_i}\varepsilon)\in\derivGp{s_{i-1}+k_i,s_i}$.
Then, by induction hypothesis, there are $\xi_i\in h^{-1}(t_i)$ and $\theta_i\in\Delta^*$ such that $(B,s_{i-1}+k_i)\Rightarrow^{\theta_i}(\xi_i,s_i)$ for each $i\in[n]$. Now we proceed with a case distinction on $\tau$:

\emph{Case 1}: Let $\tau=A\to B_1\beta_1\ldots B_n\beta_n$. Then $\beta_i=\gamma^{k_i}$ for each $i\in[n]$ and $s_0=m_1$. By construction, the transitions $\bar\tau=A\xrightarrow{\top\slash 0}\tau([B_1\gamma^{k_1}],\ldots,[B_n\gamma^{k_n}])$ and $\tau_i=[B_i\gamma^{k_i}]\xrightarrow{\top\slash k_i}B$ are in $\Delta$. But then we obtain
\begin{align*}
    (A,m_1)\Rightarrow^{\bar\tau}&(\tau([B_1\gamma^{k_1}],\ldots,[B_n\gamma^{k_n}]),s_0)\Rightarrow^{\tau_1}(\tau(B_1,\ldots,[B_n\gamma^{k_n}]),s_0+k_1)\\
    \Rightarrow^{\theta_1}&(\tau(\xi_1,\ldots,[B_n\gamma^{k_n}]),s_1)\\
    &\quad \vdots\\
    \Rightarrow^{\theta_{n-1}}&(\tau(\xi_1,\ldots,\xi_{n-1},[B_n\gamma^{k_n}]),s_{n-1})\Rightarrow^{\tau_n}(\tau(B_1,\ldots,B_n),s_{n-1}+k_n)\\
    \Rightarrow^{\theta_n}&(\tau(\xi_1,\ldots,\xi_n),s_n).
\end{align*}

\emph{Case 2}: Let $\tau=A\gamma\to B_1\beta_1\ldots B_n\beta_n$. Then $\beta_i=\gamma^{k_i}$ for each $i\in[n]$ and $s_0=m_1-1$; thus, $m_1>0$. By construction, the transitions $\bar\tau=A\xrightarrow{{>} 0\slash -1}\tau([B_1\gamma^{k_1}],\ldots,[B_n\gamma^{k_n}])$ and $\tau_i=[B_i\gamma^{k_i}]\xrightarrow{\top\slash k_i}B$ are in $\Delta$. But then we obtain
\begin{align*}
    (A,m_1)\Rightarrow^{\bar\tau}&(\tau([B_1\gamma^{k_1}],\ldots,[B_n\gamma^{k_n}]),s_0)\Rightarrow^{\tau_1}(\tau(B_1,\ldots,[B_n\gamma^{k_n}]),s_0+k_1)\\
    \Rightarrow^{\theta_1}&(\tau(\xi_1,\ldots,[B_n\gamma^{k_n}]),s_1)\\
    &\quad \vdots\\
    \Rightarrow^{\theta_{n-1}}&(\tau(\xi_1,\ldots,\xi_{n-1},[B_n\gamma^{k_n}]),s_{n-1})\Rightarrow^{\tau_n}(\tau(B_1,\ldots,B_n),s_{n-1}+k_n)\\
    \Rightarrow^{\theta_n}&(\tau(\xi_1,\ldots,\xi_n),s_n).
\end{align*}

\emph{Case 3}: Let $\tau=A\#\to B_1\beta_1\ldots B_n\beta_n$. Then $\beta_1=\gamma^{k_1}\#$, $\beta_i=\gamma^{k_i}$ for each $i\in[n]\setminus\{1\}$, and $s_0=m_1=0$. By construction, the transitions $\bar\tau=A\xrightarrow{0\slash 0}\tau([B_1\gamma^{k_1}\#],\ldots,[B_n\gamma^{k_n}])$, $\tau_1=[B_1\gamma^{k_1}\#]\xrightarrow{\top\slash k_1}B$, and $\tau_i=[B_i\gamma^{k_i}]\xrightarrow{\top\slash k_i}B$ for $i\in[n]\setminus\{1\}$ are in $\Delta$. But then we obtain
\begin{align*}
    (A,0)\Rightarrow^{\bar\tau}&(\tau([B_1\gamma^{k_1}\#],\ldots,[B_n\gamma^{k_n}]),0)\Rightarrow^{\tau_1}(\tau(B_1,\ldots,[B_n\gamma^{k_n}]),0+k_1)\\
    \Rightarrow^{\theta_1}&(\tau(\xi_1,\ldots,[B_n\gamma^{k_n}]),s_1)\\
    &\quad \vdots\\
    \Rightarrow^{\theta_{n-1}}&(\tau(\xi_1,\ldots,\xi_{n-1},[B_n\gamma^{k_n}]),s_{n-1})\Rightarrow^{\tau_n}(\tau(B_1,\ldots,B_n),s_{n-1}+k_n)\\
    \Rightarrow^{\theta_n}&(\tau(\xi_1,\ldots,\xi_n),s_n).
\end{align*}

\sepstars

Now let $\xi\in\L(\A)$. Then there is some computation $(S,0)\Rightarrow^{\theta}(\xi,m)$ in $\compA{\xi}$ for some $m\in\Nat$. If $\xi=(S\to\varepsilon)$, $m=0$ and $\theta=\bar\tau$, where $\bar\tau=(S\xrightarrow{\top\slash 0} \tau)$ with $\tau=(S\to\varepsilon)$. By construction, $\tau\in P$ and, thus, $(S\Rightarrow^{\tau}\varepsilon)\in\derivGp{\varepsilon}$. Otherwise, by construction, $\theta$ has to be of the form $\bar\tau\theta'$ with $\bar\tau=(S\xrightarrow{\top\slash 0}\tau(A))$ for some $A\in N$ and $\tau=(S\to\A\#)\in P$. Moreover, $\xi=\tau(\xi_1)$. Then $(A,0)\Rightarrow^{\theta'}(\xi_1,m)$ for some $m\in\Nat$ and, by Property \ref{prop:empt1}, $(A\#\Rightarrow^{h(\xi_1)}\varepsilon)\in\derivGp{0,m}$. But then $(S\Rightarrow^{\tau h(\xi_1)}\varepsilon)\in\derivGp{\varepsilon}$.


Now let $(S\Rightarrow^t\varepsilon)\in\derivGp{\varepsilon}$. With a similar argumentation as above and Property \ref{prop:empt2} we can argue that there is a $\xi\in h^{-1}(t)$ with $\xi\in\L(\A)$. \qed
\end{proof}

\subsection{Membership}

\textbf{Lemma\;\ref{lem:behaviour-automaton-language}.}
Let $\A = (Q, \Sigma, q_0, \Delta)$ be a GOCTA and $\xi \in T_\Sigma$.
If $\A'$ is the $(|\xi|\cdot|Q|^2+1)$-bounded behaviour automaton of $\A$, then
$\xi \in \L(\A) \iff \xi \in \L(\A')$.

\begin{proof}
    We assume that $\A$ is normalized and $0$-accepting.
    Let $m = |\xi|\cdot|Q|^2+1$ and $\A' = (Q', \Sigma, (q_0, 0, 0), \Delta')$.

    \enquote{$\xi \in \L(\A) \implies \xi \in \L(\A')$}:
    If $\xi \in \L(\A)$, then, by Lemma~\ref{lem:polybound},
    there exists $t \in \compA{\xi}$
    such that $\maxcnt(t) \le m$.
    We let $t = (\zeta_0, c_0) \Rightarrow^{\tau_1 \cdots \tau_n} (\zeta_n, c_n)$
    and we note that $\zeta_0 = q_0$, $\zeta_n = \xi$, and $c_0 = c_n = 0$.
    We will inductively define a computation of $\xi$ in $\A'$,
    denoted by $t' = \zeta_0' \Rightarrow^{\tau'_1 \cdots \tau'_n} \zeta_n'$
    such that, for every $i \in [0,n-1]$, the following invariant holds.
    If $\zeta_i = \ang{\xi}_k[q, \bar{q}]$,
    then $\zeta_i' = \ang{\xi}_k[(q,c_i,c'), \bar{p}]$ with $|\bar{q}| = |\bar{p}|$.

    We start by letting $\zeta_0' = (q_0, 0, 0)$.
    Clearly, for $i = 0$, the invariant holds.
    Now let $i \in [n]$, $\zeta_{i-1} = \ang{\xi}_k[q, \bar{q}]$, and $\zeta_{i-1}' = \ang{\xi}_k[(q,c_{i-1},c'), \bar{p}]$.
    Since $(\zeta_{i-1},c_{i-1}) \Rightarrow^{\tau_i} (\zeta_i,c_i)$,
    we need to consider two cases.
    \begin{enumerate}
        \item If $\tau_i = q \xrightarrow{\p \slash \z} q'$,
            $\zeta_i = \ang{\xi}_k[q', \bar{q}]$, $c_i = c_{i-1} + \z \ge 0$,
            and $\p(c_{i-1})$ holds,
            then we let $\tau'_i = (q, c_{i-1}, c') \to (q', c_i, c')$
            and $\zeta_i' = \ang{\xi}_k[(q', c_i, c'), \bar{p}]$.
            Moreover, $c_{i-1}, c_i, c' \le \maxcnt(t) \le m$.
            Thus, by definition of $\A'$, we have $\tau'_i \in \Delta'$.
            Clearly, $\zeta_{i-1}' \Rightarrow^{\tau'_i} \zeta_i'$.
        \item If $\tau_i = q \to \sigma(p_1, \dots, p_\ell)$,
            $\zeta_i = \ang{\xi}_{k+1}[p_1, \dots, p_\ell, \bar{q}]$, $c_i = c_{i-1}$,
            and $\sigma = \xi(\rho)$ where $\rho$ is the lexicographically $(k+1)$-st element of $\pos(\xi)$,
            then, intuitively, for each $j \in [2,\ell]$ we have to encode the counter value
            of $\A$ when deriving $p_j$ into
            the state which corresponds to $p_j$ in the computation of $\xi$ in $\A'$.
            For this, we peek at the step when $p_j$ is derived in the computation of $\xi$ in $\A$.
            Formally, for each $j \in [2, \ell]$,
            we let $w_j \in \pos(\ang{\xi}_{k+1})$ such that $\ang{\xi}_{k+1}(w_j) = x_j$
            and $u_j = \min \{i' \in [n] \mid \zeta_{i'}(w_j) = x_1 \}$.
            Then, we let
            \[
                \tau'_i = (q, c_{i-1}, c') \to \sigma((p_1, c_{i-1}, c_{u_2}), (p_2, c_{u_2}, c_{u_3}), \dots, (p_\ell, c_{u_\ell}, c'))
            \]
            and
            \[
                \zeta_i' = \ang{\xi}_{k+1}[(p_1, c_{i-1}, c_{u_2}), (p_2, c_{u_2}, c_{u_3}), \dots, (p_\ell, c_{u_\ell}, c'), \bar{p}].
            \]
            In order to show that $\tau'_i \in \Delta'$, we distinguish three cases.
            \begin{enumerate}
                \item If $\ell > 0$, then $\tau'_i \in \Delta'$ by definition of $\A'$.
                \item If $\ell = 0$ and $i < n$,
                    then $\bar{p} \not= \varepsilon$ and we let $(p,c'',c''') = \bar{p}[1]$.
                    By construction of $(p,c'',c''')$, we have $c_i = c''$
                    and $c' = c''$.
                    Thus, $c_{i-1} = c_i = c'$ and $\tau'_i \in \Delta'$ by definition of $\A'$.
                \item If $\ell = 0$ and $i = n$, then $\bar{p} = \varepsilon$.
                    Moreover, by definition of $\Delta'$, for every $i' \in [n-1]$ it holds that the third component of the state at the lexicographically largest position of $\zeta'_{i'}$ is $0$, hence $c' = 0$.
                    Since $(\zeta_{i-1}, c_{i-1}) \Rightarrow^{\tau_i} (\xi, 0)$,
                    we obtain $c_{i-1} = 0$.
                    Thus, $c_{i-1} = c_i = c'$ and $\tau'_i \in \Delta'$ by definition of $\A'$.
            \end{enumerate}
            In all three cases,
            it is obvious that $\zeta_{i-1}' \Rightarrow^{\tau'_i} \zeta_i'$.
    \end{enumerate}
    We have shown that $(q_0, 0, 0) = \zeta_1' \Rightarrow^{\tau'_1 \cdots \tau'_n} \zeta_n' = \xi$
    and hence $\xi \in \L(\A')$.

    \enquote{$\xi \in \L(\A') \implies \xi \in \L(\A)$}:
    If $\xi \in \L(\A')$, there are $n \in \Nat$,
    $\zeta_0, \dots, \zeta_n \in T_\Sigma(Q')$, and $\tau'_1, \dots, \tau'_n \in \Delta'$
    such that $\zeta_0 \Rightarrow^{\tau'_1 \cdots \tau'_n} \zeta_n$,
    $\zeta_0 = (q_0, 0, 0)$, and $\zeta_n = \xi$.
    We will inductively define a computation of $\xi$ in $\A$,
    denoted by $t' = (\zeta_0', c_0) \Rightarrow^{\tau_1 \cdots \tau_n} (\zeta_n', c_n)$
    such that, for every $i \in [0,n-1]$, the following invariant holds.
    If $\zeta_i = \ang{\xi}_k[(q,c,c'), \bar{p}]$,
    then $\zeta_i' = \ang{\xi}_k[q, \bar{q}]$
    with $|\bar{p}| = |\bar{q}|$, and $c = c_i$.

    We start by letting $\zeta_0' = q_0$ and $c_0 = 0$.
    Then, clearly, the invariant holds for $i = 0$.
    Now let $i \in [n]$, $\zeta_{i-1} = \ang{\xi}_k[(q,c_{i-1},c'), \bar{p}]$,
    and $\zeta_{i-1}' = \ang{\xi}_k[q, \bar{q}]$.
    Since $\zeta_{i-1} \Rightarrow^{\tau'_i} \zeta_i$,
    we need to consider three cases.
    \begin{enumerate}
        \item If $\tau'_i = (q,c_{i-1},c') \to (q',c_{i-1}+\z,c')$ and
            $\zeta_i = \ang{\xi}_k[(q',c_{i-1}+\z,c'), \bar{p}]$
            for some $\z \in \Nat$ with $c_{i-1} + \z \in [0, m]$,
            then, by definition of $\Delta'$,
            there exists $\p \in \{ 0, {>}0, \top \}$ such that $\p(c_{i-1})$ holds
            and $(\tau_i = q \xrightarrow{\p \slash \z} q') \in \Delta$.
            We let $\zeta_i' = \ang{\xi}_k[q', \bar{q}]$ and $c_i = c_{i-1} + \z$;
            clearly, $(\zeta_{i-1}', c_{i-1}) \Rightarrow^{\tau_i} (\zeta_i', c_i)$.
        \item If $\tau'_i = (q,c_{i-1},c') \to \sigma((p_1,c_{i-1},c_2'), (p_2,c_2',c_3'), \dots, (p_\ell,c_\ell',c'))$, \\
            $\zeta_i = \ang{\xi}_{k+1}[(p_1,c_{i-1},c_2'), (p_2,c_2',c_3'), \dots, (p_\ell,c_\ell',c'),\bar{p}]$,
            and $\sigma = \xi(\rho)$ where $\rho$ is the lexicographically $(k+1)$-st element of $\pos(\xi)$,
            then, by definition of $\Delta'$, it holds that
            $(\tau_i = q \to \sigma(p_1, \dots, p_\ell)) \in \Delta$.
            We let $\zeta_i' = \ang{\xi}_{k+1}[p_1, \dots, p_\ell, \bar{q}]$ and $c_i = c_{i-1}$;
            clearly, $(\zeta_{i-1}', c_{i-1}) \Rightarrow^{\tau_i} (\zeta_i', c_i)$.
        \item If $\tau'_i = (q,c_{i-1},c') \to \alpha$,
            $\zeta_i = \ang{\xi}_{k+1}[\bar{p}]$, $c_{i-1} = c'$,
            and $\alpha = \xi(\rho)$ where $\rho$ is the lexicographically $(k+1)$-st element of $\pos(\xi)$,
            then, by definition of $\Delta'$, it holds that
            $(\tau_i = q \to \alpha) \in \Delta$.
            We distinguish two cases.
            \begin{enumerate}
                \item If $i = n$, then $\ang{\xi}_{k+1} = \xi$ and $\bar{q} = \bar{p} = \varepsilon$.
                    Since $(q_0,0,0) \Rightarrow^{\tau'_1 \cdots \tau'_n} \xi$,
                    we have $c' = 0$ and hence $c_{i-1} = 0$.
                    Clearly, $(\zeta_{i-1}',0) \Rightarrow^{\tau_i} (\xi,0)$.
                \item If $i < n$, then $\bar{p} \not= \varepsilon$
                    and we let $(p, c'', c''') = \bar{p}[1]$.
                    Moreover, there exists $i' \in [i-1]$, $k' \in [0,k-1]$,
                    $p'', p''' \in Q$, and $d, d' \in [0,m]$
                    such that
                    \begin{align*}
                        \ang{\xi}_{k'}[(p'',c_{i'-1},d), \bar{p'}] 
                        &\Rightarrow^{\tau'_{i'}}
                        \ang{\xi}_{k'+1}[\dots, (p''',d',c'), (p',c'',c''') \dots, \bar{p'}] \\
                        &\Rightarrow^{\tau'_{i'+1} \cdots \tau'_{i_1}}
                        \ang{\xi}_k[(q,c_{i-1},c'), \bar{p}]
                    \end{align*}
                    and $c' = c''$; thus, also $c_{i-1} = c''$.
                    We let $\zeta_i' = \ang{\xi}_{k+1}[\bar{q}]$ and $c_i = c_{i-1}$;
                    clearly, $(\zeta_{i-1}', c_{i-1}) \Rightarrow^{\tau_i} (\zeta_i', c_i)$.
            \end{enumerate}
    \end{enumerate}
    We have shown that $(q_0, 0) = (\zeta_1', c_0) \Rightarrow^{\tau_1 \cdots \tau_n} (\zeta_n', c_n) = (\xi, 0)$
    and hence $\xi \in \L(\A)$.\qed
\end{proof}

\end{document}

%% file: notation.tex
\newcommand{\Nat}{\mathbb{N}}
\newcommand{\Int}{\mathbb{Z}}
\newcommand{\pos}{\mathrm{pos}}
\newcommand{\Run}{\mathrm{Run}}
\newcommand{\cnt}{\mathrm{cnt}}
\renewcommand{\inf}{\mathrm{inf}}
\newcommand{\fin}{\mathrm{fin}}
\newcommand{\PC}{\mathrm{PC}}
\newcommand{\A}{\mathcal{A}}
\newcommand{\G}{\mathcal{G}}
\newcommand{\Q}{{}Q{}}
\newcommand{\F}{\mathcal{F}}
\newcommand{\V}{\mathcal{V}}
\newcommand{\Pow}{\mathcal{P}}
\newcommand{\Vfo}{\mathcal{V}_\mathrm{fo}}
\newcommand{\Vso}{\mathcal{V}_\mathrm{so}}
\renewcommand{\L}{\mathcal{L}}
\newcommand{\NF}{\mathrm{NF}}
\newcommand{\PDNF}{\mathrm{PDNF}}
\newcommand{\PCNF}{\mathrm{PCNF}}
\newcommand{\MSO}{\mathrm{MSO}}
\newcommand{\MSOP}{\mathrm{MSO}^{\mathrm{wParikh}}}
\newcommand{\MSOeP}{\mathrm{MSO}^{\exists\mathrm{wParikh}}}
\newcommand{\To}{T^\omega}
\renewcommand{\rho}{\varrho}
\newcommand{\free}{\mathrm{free}}
\newcommand{\proj}{\mathrm{pr}}
\newcommand*\yield{\mathrm{yield}}
\newcommand*\lhs{\mathrm{lhs}}
\newcommand*\leaves{\mathrm{leaves}}
\newcommand*\maxcnt{\mathrm{maxcnt}}
\newcommand*\rk{\mathrm{rk}}
\newcommand*\size{\mathrm{size}}
\newcommand*\height{\mathrm{ht}}
\newcommand*\val{\mathrm{val}}

\newcommand{\ang}[1]{\langle#1\rangle}
\newcommand{\width}[1]{{\mathrm{w}(#1)}}
\newcommand{\confA}{\mathrm{conf}(\A)}
\newcommand{\compA}[1]{\mathrm{comp}_\A(#1)}
\newcommand{\derivG}[1]{\mathrm{deriv}_G(#1)}
\newcommand{\derivGp}[1]{\mathrm{deriv}_{G'}(#1)}
\newcommand{\derivGall}{\mathrm{deriv}(G)}
\newcommand{\derivGpall}{\mathrm{deriv}(G')}
\newcommand{\compAp}[1]{\mathrm{comp}_{\A'}(#1)}
\newcommand{\diff}{\mathrm{diff}}
\newcommand{\cp}{\mathrm{copy}}
\newcommand{\p}{\mathscr{p}}
\newcommand{\z}{\mathscr{z}}

\newcommand{\lhe}[1]{\color{purple}{#1}\color{black}}

\newcommand{\sepstars}{\bigskip\par\centerline{*\,*\,*}\medskip\par}%

\newcommand\xqed[1]{%
  \leavevmode\unskip\penalty9999 \hbox{}\nobreak\hfill
  \quad\hbox{#1}}
\newcommand{\exqed}{\xqed{$\triangleleft$}}

\newcommand*\lex{\leq_{\mathrm{lex}}}

%% file: lit.bib
@book{tata,
  TITLE = {{Tree Automata Techniques and Applications}},
  AUTHOR = {Comon, H. and Dauchet, Max and Gilleron, R{\'e}mi and Jacquemard, Florent and Lugiez, Denis and L{\"o}ding, Christof and Tison, Sophie and Tommasi, Marc},
  URL = {https://inria.hal.science/hal-03367725},
  PAGES = {262},
  YEAR = {2008},
  HAL_ID = {hal-03367725},
  HAL_VERSION = {v1},
}

@article{DusMidPar92,
	title = {Indexed counter languages},
	volume = {26},
	doi = {10.1051/ita/1992260100931},
	number = {1},
	journal = {RAIRO - Theoretical Informatics and Applications},
	author = {Duske, J. and Middendorf, M. and Parchmann, R.},
	year = {1992},
	pages = {93--113}
}

@misc{AmaJea13a,
	title = {A {Proof} of the {Pumping} {Lemma} for {Context}-{Free} {Languages} {Through} {Pushdown} {Automata}},
	archivePrefix = {arXiv},
	author = {Amarilli, A. and Jeanmougin, M.},
	year = {2013},
	eprint = {1207.2819},
}

@phdthesis{Castano04,
  title={Global index languages},
  author={Castano, J.},
  year={2004},
  school={Brandeis University},
  url={www.cs.brandeis.edu/~jcastano/thesis3.pdf}
}

@phdthesis{Kla04,
  author    = {F. Klaedtke},
  title     = {Automata-based decision procedures for weak arithmetics},
  school    = {University of Freiburg},
  year      = {2004},
  url       = {http://freidok.ub.uni-freiburg.de/volltexte/1439/index.html}
}

@incollection{Gue81,
	address = {Berlin, Heidelberg},
	title = {On pushdown tree automata},
	volume = {112},
	booktitle = {{CAAP} '81},
	publisher = {Springer Berlin Heidelberg},
	author = {Guessarian, I.},
	editor = {Goos, G. and Hartmanis, J. and Brauer, W. and Brinch Hansen, P. and Gries, D. and Moler, C. and Seegmüller, G. and Stoer, J. and Wirth, N. and Astesiano, Egidio and Böhm, Corrado},
	year = {1981},
	doi = {10.1007/3-540-10828-9_64},
	pages = {211--223}
}

@incollection{KlaRue03,
	title = {Monadic {Second}-{Order} {Logics} with {Cardinalities}},
	volume = {2719},
	booktitle = {Automata, {Languages} and {Programming}},
	publisher = {Springer},
	author = {Klaedtke, F. and Rueß, H.},
	year = {2003},
	doi = {10.1007/3-540-45061-0_54},
	pages = {681--696}
}

@techreport{Eng86,
  type = {Technical Report},
  title = {Context-Free Grammars with Storage},
  author = {Engelfriet, J.},
  year = {1986},
  number = {86-11},
  institution = {University of Leiden},
  archivePrefix = {arXiv},
  eprint = {1408.0683}
}
